\documentclass[reqno,11pt,a4paper]{amsart}

\usepackage{amssymb}

\usepackage{enumitem}

\usepackage{amsmath}
\usepackage{mathabx}

\usepackage{color}
\definecolor{darkred}{rgb}{0.9,0.1,0.1}

\pdfpkmode={supre}
\usepackage{hyperref}
\hypersetup{
    colorlinks,
    citecolor=red,
    filecolor=black,
    linkcolor=blue,
    urlcolor=blue
}

\usepackage{fancyhdr} 

\newcommand{\changefont}{
    \fontsize{9}{11}\selectfont
}
\newtheorem{theorem}{Theorem}[section]

\newtheorem{proposition}[theorem]{Proposition}
\newtheorem{remark}[theorem]{Remark}
\newtheorem{lemma}[theorem]{Lemma}
 \newtheorem{corollary}[theorem]{Corollary}

      \newtheorem{assumption}{}

\theoremstyle{definition}

\newcommand{\mb}[1]{\mathbb{#1}}
\newcommand{\mc}[1]{\mathcal{#1}}

\newcommand{\mbf}[1]{\mathbf{#1}}

\newcommand{\lp}{\langle}
\newcommand{\rp}{\rangle}

\newcommand{\vp}{\varphi}
\newcommand{\ve}{\varepsilon}

\newcommand{\im}{\text{i}}

 \DeclareMathOperator{\spec}{Spec}

\DeclareMathOperator{\hess}{Hess}
\DeclareMathOperator{\supp}{supp}

\author[G. Di Ges\`u]{Giacomo Di Ges\`u}
\address{ Giacomo Di Ges\`u, Univerist\`a di Pisa, 
Largo Bruno Pontecorvo 5, 
56127 Pisa, Italy.}
\email{giacomo.digesu@unipi.it}

\title{Spectral analysis of discrete metastable diffusions}

\begin{document}



\begin{abstract}
\noindent
We consider a discrete  Schr\"odinger operator 
$ H_\ve= -\ve^2\Delta_\ve + V_\ve$ on $\ell^2(\ve \mb Z^d)$, where $\ve>0$ is a small parameter 
and the potential $V_\ve$ is defined in terms of a multiwell energy landscape $f$ on $\mb R^d$.  
This operator can be seen as a
discrete analog of the semiclassical Witten Laplacian
of $\mb R^d$. 
It is unitarily equivalent to the generator 
of a diffusion on $\ve \mb Z^d$, 
satisfying the detailed balance condition with respect to the Boltzmann weight $\exp{(-f/\ve)}$.      
These type of diffusions exhibit metastable behavior and arise in the context of disordered mean field models in Statistical Mechanics. 
We analyze 
  the bottom of the spectrum of $H_\ve$ in the semiclassical regime $\ve\ll1$
and show that there is a one-to-one correspondence between exponentially small eigenvalues and local minima of $f$. Then we analyze in more detail the bistable case    
and compute the precise 
asymptotic splitting between the two exponentially small eigenvalues. 
Through this purely spectral-theoretical analysis of the     discrete Witten Laplacian 
we recover in a self-contained way the Eyring-Kramers formula for the metastable tunneling time of the underlying stochastic process. 
\end{abstract}

\maketitle

\section{Introduction}

\noindent
This paper derives sharp semiclassical spectral asymptotics for Schr\"odinger operators acting on $\ell^2(\ve\mb Z^d)$ of the form 
\begin{equation}\label{introH}       H_\ve    =     - \ve^2 \Delta_\ve    +   V_\ve         ,    \  \  \     0<\ve\ll1  ,  
\end{equation}
where $\Delta_\ve$ is the discrete nearest-neighbor Laplacian of $\ve \mb Z^d$ and $V_\ve$ is a possibly unbounded multiplication operator, defined in terms
of a 
multiwell energy  landscape $f$.
More precisely, given $f\in C^2(\mb R^d)$, we identify $V_\ve$
with the function 
\begin{equation}   \label{introV}  V_\ve (x)     =    e^{\frac{f(x)}{2\ve}} (\ve^2 \Delta_\ve e^{-\frac{f}{2\ve}} )(x)      . 
\end{equation}  
We shall dub $H_\ve$ the discrete semiclassical Witten Laplacian associated with $f$. This is motivated by the following observation: 
the continuous space version of $H_\ve$, i.e. the Schrödinger operator $\mc H_\ve$ on $L^2(\mb R^d)$ 
obtained from~\eqref{introH},\eqref{introV} by substituting $\Delta_\ve$ with the Laplacian $\Delta$ of $\mb R^d$, reads 
\begin{equation}\label{introHcont}
\mc H_\ve    =   - \ve^2 \Delta +   \tfrac 14 |\nabla f|^2 - 
 \tfrac{\ve}{2} \Delta f,
\end{equation}
and thus coincides with the restriction on functions of the Witten Laplacian of $\mb R^d$ \cite{Witten, HSIV, HN(LNM), Miclo}. It is well known that the latter has deep connections to problems in 
Statistical Mechanics~\cite{H(WS)}. In some situations, e.g. when considering lattice models of Statistical Mechanics as discussed below, one is 
led in a natural way 
to its discrete version~\eqref{introH},\eqref{introV}. 
The continuous space operator~\eqref{introHcont} is then rather a simplifying idealization of~\eqref{introH},\eqref{introV}: it is indeed easier to analyze $\mc H_\ve$ by exploiting the standard machinery of differential and semiclassical calculus, but the results might be a priori less accurate in making predictions. This paper shows a general
strategy which permits to obtain sharp semiclassical estimates directly in the discrete setting. 
 
\

We are mainly inspired by the analysis~\cite{HKN} on the continuous space Witten Laplacian and by the series of papers~\cite{KleinRosenbergerI, KleinRosenbergerII, KleinRosenbergerIII, KleinRosenbergerIV} by M. Klein and E. Rosenberger, 
who develop  an approach to 
 the semiclassical spectral analysis of discrete Schr\"odinger operators of the form~\eqref{introH} via microlocalization techniques. We refer also to the earlier work~\cite{HSharper} and 
 to~\cite{colinpan,colindiscrete} for semiclassical investigations in discrete settings.

\

\

\noindent
{\bf Brief description of the main results.}  

\,

\noindent
Following in particular the approach of~\cite{KleinRosenbergerII} we show that 
under mild regularity assumptions on $f$ there is a low-lying spectrum of exponentially small eigenvalues which is well separated from the rest of the spectrum. Moreover the number of exponentially small eigenvalues equals the number of local minima of $f$, see Theorem~\ref{main1F} below. 

\

Then we analyze in more detail the case of two local minima of $f$
and compute the precise 
asymptotic splitting between the two small eigenvalues. 
From a general point of view, this corresponds to a subtle 
tunneling calculation through other, non-resonant wells ~\cite{HSIII} of the Schr\"odinger potential $V_\ve$, corresponding to saddle points of $f$.

 \

  As opposed to~\cite{HKN} we work again under mild regularity assumptions on $f$ and proceed with a streamlined, direct strategy that avoids WKB expansions, a priori Agmon estimates and also the underlying complex structure of the Witten Laplacian.  Much of the simplification is obtained via a suitable choice of 
global quasimodes. 
We show that the leading asymptotic of the exponentially small eigenvalue gap is given by an Eyring-Kramers formula: \[    \lambda(\ve )         =          
   \ve    A   e^{-\frac{E}{\ve}}    (1 + o(1))     
     ,    \]
where $A,E>0$ are explicit constants depending on $f$ (see Theorem~\ref{spectralgapF} for a precise statement) that turn out to coincide with the one obtained  
in the continuous 
case for $\mc H_\ve$ in~\cite{HKN} (see also~\cite{BGK, Eckhoff}). In other terms, 
the geometric constraint imposed by the lattice turns out to be negligible in first order approximation. The vanishing rate of the remainder term depends on the regularity of $f$ around its critical points. We show  that $f\in C^3(\mb R^d)$ 
implies an error of order $O(\sqrt\ve)$. 
\

The spectral Eyring-Kramers formula in the discrete setting considered here is not new. Indeed, up to some minor variants, this type of result has been derived in the framework of discrete metastable diffusions, by analyzing mean transition times 
of Markov processes via potential theory~\cite{BEGKdiscII}. We shall discuss below more in detail the probabilistic interpretation of our results. 
The present paper 
shows  that, as in the continuous setting, also 
 in the discrete setting the 
 Eyring-Kramers formula can be obtained 
by a direct and self-contained 
spectral approach, without relying at all on probabilistic potential theory.

\

We remark that the method we use to analyze the exponentially small eigenvalues 
can be extended also to the general case with more than two local minima. The extension is based on an iterative finite-dimensional matrix procedure, very similar to the one considered in~\cite{HKN} (see also~\cite{thesis} and references therein). This procedure is independent of the rest and not related to the peculiar analytical difficulties arising from the discrete character of the setting. To not 
 obscure the exposition of the 
main ideas of this paper, the general case will be discussed somewhere else.

\

\

\noindent{\bf Connection to discrete metastable diffusions.}

\,

\noindent
Our main motivation for investigating the spectral properties of $H_\ve$ 
stems from its close connection to certain metastable diffusions with state space $\ve \mb Z^d$. These have been extensively studied in the probabilistic literature, mainly due to their paradigmatic properties and their applications to problems in Statistical Mechanics~\cite{cassandro, BEGKdiscII, bianchi, DLPeres,BeltranLandim, LMT, SS}. 
The general, continuous time version might be described in terms 
of a  
Markovian generator $L_\ve$
of the form 
\begin{equation}\label{introL}       L_{\ve} \psi(x)  =      \sum_{v\in \mb Z^d}      
      r_\ve (x, x+\ve v)   \left[ \psi(x+ \ve v )      -    \psi(x )   \right]
    ,     
\end{equation}
with $r_\ve (x, x+\ve v)$ being the rate of a jump from $x$ to $x+\ve v$.
The jump rates are assumed to satisfy the detailed balance condition with respect to the Boltzmann weight $\rho_\ve = e^{-f/\ve}$ on $\ve\mb Z^d$, so that $L_{\ve}$ 
may be realized as a selfadjoint operator acting on the weighted space $\ell^2(\ve \mb Z^d; \rho_\ve)$.  Moreover the scaling is chosen so that $ L_\ve$ formally converges for $\ve\to 0$ to a first order differential operator on $\mb R^d$, corresponding to a deterministic transport along a  vector field. One might thus think of the dynamics as a small stochastic perturbation of a deterministic motion. 
A standard choice of jump rates satisfying the above requirements is given by 
\begin{equation}\label{introRates} 
   r_\ve(x, x+ \ve v)   =   \begin{cases}   \tfrac {1}{\ve}e^{- \tfrac {1}{2\ve} [ (f(x+\ve v) - f(x)]}       &     \text{ if }     v\in \{-e_k, e_k\}_{k=1,\dots, d},  \\
       0      & \text{ otherwise} ,  
   \end{cases}
\end{equation}
where $(e_1, \dots, e_d)$ is the standard basis of $\mb R^d$. 

\

There is a direct link between the discrete Witten Laplacian and 
discrete diffusions as described above: 
up to a change of sign and multiplicative factor $\ve$,
the Markovian generator $L_\ve$ given by~\eqref{introL},\eqref{introRates} and the discrete Witten Laplacian given by~\eqref{introH},\eqref{introV} are formally unitarily equivalent.
This can be seen by the well-known ground state transformation, which turns a Schr\"odinger operator into a diffusion operator~\cite{JLMS}, see Proposition~\ref{GST} below for the precise statement. As a consequence, our spectral analysis of
$H_\ve$ can be immediately translated into analogous results on
$L_\ve$, see Corollary~\ref{MainCorollary}. The advantage of working with $H_\ve$
is that in the flat space $\ell^2(\ve \mb Z^d)$ one can exploit Fourier analysis and related microlocalization techniques.

\

We remark that discrete diffusions as described above naturally arise in the context of disordered mean field models in Statistical Mechanics.
A prominent example is the dynamical random field Curie-Weiss model~\cite{FMP, BEGKdiscII, bianchi, SS}, which is well described by a discrete diffusion on $\ve \mb Z^d$ after a suitable reduction in terms of order parameters.
The limit $\ve\to 0 $ then corresponds to the thermodynamic limit of infinite volume.

\

A characteristic feature of the dynamics $\partial_t \psi = L_\ve \psi$
for small $\ve$ is metastability: if $f$ admits several local minima the 
system remains trapped for exponentially large times in neighborhoods of local minima of $f$ before exploring the whole state space. This is due to the fact that the local minima of $f$ turn out to be exactly the stable equilibrium points of the limiting deterministic motion. 
 We refer to \cite{FW, OV, BdH} for comprehensive introductions to metastability of Markov processes and e.g. to~\cite{berglund, BODOTO, Landimreview} for shorter surveys.

A key issue in the understanding of metastability is to quantify 
the time scales at which metastable transitions between local minima occur. For discrete diffusions of type~\eqref{introL} sharp asymptotic estimates 
have been obtained in~\cite{BEGKdiscI, BEGKdiscII} in terms of average hitting times. The formula for the leading asymptotics is called
Eyring-Kramers formula. 
In~\cite{BEGKdiscII} it is also shown that  there is a very clean relationshp between the
metastable transition times and the low-lying spectrum of $-L_\ve$.
Indeed, there is a cluster of exponentially small eigenvalues, each one being asymptotically equivalent to the inverse 
of a metastable transition time.

The problem of determining the asymptotic behavior of metastable transition times can therefore be equivalently phrased as a problem of spectral asymptotics of the generator $L_\ve$ and thus of $H_\ve$. Due to these facts, one can 
view the method presented in this paper as a spectral approach to the computation of metastable transition times in discrete setting.

\

\noindent{\bf Plan of the paper.} In Section~\ref{SectionResults} we introduce the setting, provide precise definitions and basic properties for the discrete Witten Laplacian $H_\ve$, the diffusion generator $L_\ve$ and state our main results: 
Theorem~\ref{main1F}, saying that there are as many exponentially small eigenvalues  of $H_\ve$ as minima of $f$ and that there is a large gap of order $\ve$ between them and the rest of the spectrum; Theorem~\ref{spectralgapF}, giving the precise splitting between exponentially small eigenvalues due to the tunnel effect (Eyring-Kramers formula). 
In Section~\ref{SectionTools} we collect some preliminary tools which can be seen as general means for a semiclassical analysis on the lattice: the IMS formula for the discrete Laplacian which permits to localize quadratic forms
on the lattice; estimates on the discrete semiclassical Harmonic oscillator based on microlocalization techniques; and results on sharp Laplace asymptotics on the lattice $\ve \mb Z^d$ based on the Poisson summation formula. 
In Section~\ref{SectionMainProof1} and Section~\ref{SectionMainProof2} we provide the proofs of Theorem~\ref{main1F} and Theorem~\ref{spectralgapF}  respectively.

\

\section{Precise setting and main results}\label{SectionResults}

\noindent
Throughout the paper we shall use the following notation. We consider the symmetric set
\[ \mc N  =   \{e_k, - e_k:  k=1,\dots, d \} \subset \mb Z^d ,    \]
 where $(e_1, \dots, e_d)$ is the standard basis of $\mb R^d$. 
For $\ve>0$ the symbols  $\nabla_\ve$ and   $\Delta_\ve$ denote respectively the rescaled  discrete gradient and the rescaled discrete Laplacian of the lattice 
$\ve \mb Z^d$, with graph structure induced by $\ve \mc N$. More precisely,  for every   $\psi : \ve \mb Z^d \to \mb R$     we define 
\begin{equation*}
\nabla_\ve \psi    \,  (x,v)   =     \ve^{-1} \left[\psi(x+ \ve v)  - \psi(x) \right]         ,    \    \ \      \forall    x\in \ve \mb Z^d \text{ and } v\in \mc N                     ,
\end{equation*}
\begin{equation*}      \Delta_\ve \psi   \, (x)       =         \ve^{-2}      \sum_{v\in \mc N}\left[  \psi (x+\ve v)  -    \psi(x)     \right]      ,  \  \   \
  \forall    x\in \ve \mb Z^d        .
\end{equation*}
We shall work on the Hilbert space $\ell^2(\ve \mb Z^d) =\{ \psi\in  \mb R^{\ve \mb Z^d} :  \| \vp \|_{\ell^2(\ve \mb Z^d)}  < \infty\}$, 
where 
 $\|\cdot\|_{\ell^2(\ve \mb Z^d)}$ is the norm corresponding to the 
scalar product 
\begin{equation*}     \lp \psi, \psi'\rp_{\ell^2(\ve \mb Z^d)}        =            \ve^d    \sum_{x\in \ve \mb Z^d }    \psi(x) \, \psi'(x)    .    
\end{equation*}  
The discrete Laplacian $\Delta_\ve$ is a bounded linear operator on $\ell^2(\ve \mb Z^d)$. It is also selfadjoint and $-\Delta_\ve$ is nonnegative. More precisely,  for  
 $\psi, \psi'\in \ell^2(\ve \mb Z^d)$, once can check that 
\[     \lp  -   \Delta_\ve     \psi,   \psi'\rp_{\ell^2(\ve \mb Z^d)}      =        \lp       \psi,   -   \Delta_\ve    \psi'\rp_{\ell^2(\ve \mb Z^d)}    =    
     \lp \nabla_\ve  \psi , \nabla_\ve  \psi'\rp_{\ell^2(\ve \mb Z^d; \mb R^\mc N)}          ,                     \]
and in particular
\[     \lp  -   \Delta_\ve     \psi,   \psi\rp_{\ell^2(\ve \mb Z^d)}      =         
      \| \nabla_\ve  \psi \|^2_{\ell^2(\ve \mb Z^d; \mb R^\mc N)}   \geq 0        .                    \]
Here $\|\cdot\|_{\ell^2(\ve \mb Z^d;\mb R^\mc N)}$ is the norm induced by the scalar product  
\begin{equation*}   \lp \alpha, \alpha'\rp_{\ell^2(\ve \mb Z^d;\mb R^\mc N)}     =            \frac{\ve^d}{2}    \sum_{x\in \ve \mb Z^d } \sum_{v\in \mc N}    \alpha(x,v) \, \alpha'(x,v)         ,   
 \end{equation*}  
defined for $\alpha, \alpha'\in \ell^2(\ve \mb Z^d;\mb R^\mc N) 
  :=\{ \alpha \in  \mb R^{\ve \mb Z^d \times \mc N} :  \| \alpha(\cdot, v) \|_{\ell^2(\ve \mb Z^d)}  < \infty \text{ forall } v\in \mc N\}  $ 
(the space of square integrable $1$-forms on the graph $\ve \mb Z^d$).

\

\subsection{ Definition and basic properties of $H_\ve$}

\

\noindent
Given a function $f:\mb R^d\to \mb R$ and a parameter $\ve>0$, we define a new function $ V_\ve :\mb R^d\to \mb R$ by setting 
\begin{equation}\label{mainV}     V_\ve (x)     =       \sum_{v\in  \mc N}  \big[    \    e^{- \tfrac 12 
\nabla_\ve f (x, v)  }      
-    1       \big]        ,   \  \   \  \forall x\in \mb R^d     .   
\end{equation}
 Note that the expression \eqref{mainV}  for $V_\ve$  and the one given in the introduction in~\eqref{introV} are equal by definition of $\Delta_\ve$ and $\nabla_\ve$. 
We shall identify in the sequel $V_\ve$ with the corresponding multiplication operator in $\ell^2(\ve \mb Z^d)$ having dense domain 
${\textrm Dom}(V_\ve) = \{    \psi \in \ell^2(\ve \mb Z^d):   V_\ve \psi \in \ell^2(\ve \mb Z^d)     \}$. 
The restriction of $  V_\ve$
to  $C_c(\ve \mb Z^d)$ (i.e. the set of $\psi \in \mb R^{\ve\mb Z^d}$ such that $\psi(x)=0$ for all but finitely many $x$)  is essentially selfadjoint.

\

\noindent
We are interested in the  Schr\"odinger-type operator $H_\ve :  {\textrm Dom}(V_\ve)   \to   \ell^2(\ve \mb Z^d)  $ given by 
\[       H_\ve   =         -\ve^2 \Delta_\ve    +    V_\ve                  .       \]
Note that $H_\ve$ is a selfadjoint operator in $ \ell^2(\ve \mb Z^d)  $ and 
its restriction to $C_c(\ve\mb Z^d)$ is essentially selfadjoint. This follows e.g.
from the Kato-Rellich Theorem~\cite[Theorem 6.4]{Teschl}, using the 
analogous properties of $V_\ve$ and the fact that $\Delta_\ve  $ is bounded and selfadjoint. 

Moreover, from the pointwise bound $V_\ve \geq - 2d$ and the nonnegativity of $-\Delta_\ve$ it follows immediately that $H_{\ve}$ is bounded from below. 
An important observation is that the quadratic form associated with $H_{\ve}$ is not only bounded from below, but even nonnegative. This is due to the special form 
of the potential $V_\ve$. Indeed, a straightforward computation yields    
 \begin{equation}\label{Witten*}    \lp H_{\ve}    \psi , \psi   \rp_{\ell^2(\ve\mb Z^d)} 
         =            \|  \nabla_{f,\ve} \psi  \|^2_{\ell^2(\ve \mb Z^d; \mb R^\mc N)}  \geq 0       ,   \   \     \    \forall    \psi \in  {\textrm Dom}(V_\ve)                ,      
\end{equation} 
where $\nabla_{f, \ve}$ denotes a suitably weighted discrete gradient: 
\begin{equation*}
\nabla_{f,\ve}  \psi    \,  (x,v)   =        \ve    e^{-\frac{f(x)+ f(x+\ve v)}{4\ve}}   \nabla_\ve (e^{\frac{f}{2\ve}}\psi) \, (x,v)   ,           \   \   \        \    \ \  
   \forall    x\in \ve \mb Z^d \text{ and } v\in \mc N      .      
   \end{equation*}
It follows in particular that the spectrum of $H_\ve$ is contained in $[0,\infty)$.

 \

  \begin{remark}
The property~\eqref{Witten*} states that $H_\ve$ is the Laplacian associated to the distorted
gradient $\nabla_{f,\ve}$. As it is done for the continuous space Witten Laplacian~\cite{Witten, HSIV}, it is possible to give an extension of $H_\ve$ in the sense of Hodge theory.The extended operator is then defined on a suitable algebra of discrete differential forms and satisfies the usual intertwining relations. We shall not use this fact and refer to~\cite{thesis} for details. 
\end{remark}

\

\subsection{Assumptions and main results}

\

\noindent
We shall consider the following two sets of hypotheses on the function  $f$. 
Here and in the following $|\cdot|$ denotes  the standard euclidean norm on $\mb R^d$. The gradient and Hessian of a function on $\mb R^d$
are denoted by $\nabla$ and $\hess$. 
\begin{assumption} \label{H1F}$f\in C^3(\mb R^d)$ and all its critical points are nondegenerate. Moreover 
\begin{itemize}
\item[(i)]   $\liminf_{|x| \to \infty } |\nabla f (x)| >0 $.
\item[(ii)]   $\hess f$ is bounded on $\mb R^d$.
\end{itemize}
\end{assumption}
\noindent
Note that~\ref{H1F} implies that the set of critical points of $f$ is finite. Indeed, nondegenerate critical points are necessarily isolated and by~(i) the critical points of $f$ must be contained in a compact subset of $\mb R^d$. 

\

\noindent
To analyze the exponential splitting between small eigenvalues we will 
assume for simplicity the following more restrictive hypothesis.

\begin{assumption}  \label{H2}
Hyptohesis~\ref{H1F} holds true. Moreover
\begin{itemize}
\item[(i)]     $\liminf_{|x| \to \infty }   \frac{f(x)}{|x|} > 0 $.
\item[(ii)]    The function $f$ has exactly two local minimum points $m_0, m_1\in \mb R^d$.      
\end{itemize}
\end{assumption}

\noindent
The first result we present shows that under Assumption~\ref{H1F} the essential spectrum
of $H_\ve$, denoted by $\spec_{\textrm{ess}} (H_{\ve})$, is uniformly bounded away from zero
and that its discrete spectrum, denoted by $\spec_{\textrm{disc}} (H_{\ve})$, is well separated into two parts: one consists of exponentially small eigenvalues, the other of eigenvalues which are at least at distance of order $\ve$
from zero. Moreover the rank of the spectral projector corresponding to the exponentially small eigenvalues equals exactly the number of local minima of $f$: 
\begin{theorem}\label{main1F} Assume~\ref{H1F} and denote by $N_0\in \mb N_0$ the number of local minima of $f$. There exist 
constants $\ve_0 \in (0, 1)$ and $C>0$ such that
for each $\ve\in (0, \ve_0]$ the following properties hold true.
\begin{itemize}

\item[(i)]        $  \spec_{\textrm{ess}} (H_{\ve})        \subset [C, \infty) $   .

\item[(ii)]    $    |  \spec_{\textrm{disc}} (H_{\ve})   \cap [0, C\ve]   |        \leq N_0$.

\item[(iii)]   $H_\ve$ admits at least $N_0$ eigenvalues counting multiplicity. In the nontrivial case that 
$N_0\neq 0$, the $N_0$-th eigenvalue $\lambda_{N_0}(\ve)$ (according to increasing order and counting multiplicity) satisfies the bounds
\[0\leq \lambda_{N_0}(\ve)   \leq    e^{-C/\ve}.  \]
\end{itemize}
\end{theorem}
\noindent
The properties stated in Theorem~\ref{main1F} are well-known in the continous space setting~\cite{Si83,HSI} and have also been recently extended to certain infinite-dimensional situations~\cite{Brooks}. In the finite-dimensional continuous space setting the standard proof consists in approximating the Schr\"odinger operator with harmonic oscillators around the critical points of $f$. The error is then estimated using the IMS localization formula, which permits to connect the local estimates around the critical points 
to global estimates. The discrete case is analytically more difficult, due to the nonlocal character of the discrete Laplacian. The main idea to overcome these difficulties is
taken from~\cite{KleinRosenbergerII} and consists in localizing not only the potential $V_\ve$ 
but the full operator $H_\ve$. This amounts in localizing the symbol in phase space and is also referred to as micolocalization. The setting in~\cite{KleinRosenbergerII} is very general and requires  the machinery of pseudodifferential operators, which makes the proof rather involved and requires strong regularity assumptions on the potential $V_\ve$ which are not assumed here. 
Here we give a more elementary proof which is adapted to our special case and works well 
under Hypthesis~\ref{H1F}.

\

\noindent
We now assume the stronger Hypothesis~\ref{H2}. Then, thanks to the superlinear growth condition~\ref{H2} (i), it holds
\begin{equation*}
        \| e^{- \frac{f}{2\ve}}    \|_{\ell^2(\ve \mb Z^d)}     <     \infty  ,   \  \   \    \forall \ve>0    . 
\end{equation*}
This implies that
$e^{- \frac{f}{2\ve}}  $ is in the domain of $H_{\ve}$ and therefore, since 
$H_{ \ve} e^{- \frac{f}{2\ve}}  =0$ by direct computation, that $0$ is an eigenvalue of $H_{\ve}$. 
Moreover, due to the fact that $\mc N$ generates the group $\mb Z^d$, it follows for example from~\eqref{Witten*} that 
only multiples of $\Psi_\ve$ can be eigenfunctions corresponding to the eigenvalue $0$.  Thus we conclude that 
$0$ is an eigenvalue with multiplicity $1$ for every $\ve>0$.

Since, by assumption, there are $N_0=2$ local minima of $f$, it follows from Theorem~\ref{main1F} that, for $\ve>0$ sufficiently small, there is exactly one eigenvalue 
$\lambda_\ve$ of $H_{\ve}$, which is different from $0$ and is exponentially small in $\ve$. Moreover, by the same theorem, $\lambda_\ve$ must have multiplicity $1$. Our second 
main result provides the precise leading asymptotic behavior of $\lambda(\ve)$. 
This behavior is expressed in terms of two constants $A,E>0$, giving respectively the
prefactor and the exponential rate. More precisely one defines 
\begin{equation}   \label{exprate}
  E           :=       h^*      -       h_*      , 
\end{equation}
where $h_*  :=  \min\{f(m_0), f(m_1)\}\in \mb R$  is the lowest energy level and where $h^*\in \mb R$ is given by the height of the barrier which separates the two minima. More precisely, $h^* $ can be defined 
as follows ~\cite{HKN}. For $h\in \mb R$ we denote by $\mc S_f(h) :=    f^{-1} \left(  (-\infty, h) \right)$ the (open) sublevel set of $f$ corresponding to the height $h$
and by $ N_{f}(h)$ the number of connected components of $\mc S_f(h) $. 
Then $h^*(f)\in \mb R$ is defined as the maximal height which disconnects $\mc S_f(h)$ into two components:
\begin{equation}\label{heightofenergyF}    h^*     :=            \max \left\{     h\in \mb R :     N_f(h)      =        2      \right\}           .  
\end{equation}
By simple topological arguments, on the level set $f^{-1}(h^*)$ there must be at least one critical point of $f$ of index $1$ and at most a finite number $n$ of them, which we label 
in an arbitrary order as $s_1, \dots, s_{n}$.   We denote by $\mu(s_k)$ the only negative eigenvalue of $\hess f(s_k)$. The constant $A$ is then defined in terms of the quadratic curvature of $f$ around the two minima and the relevant saddle points. More precisely,
one defines 
\begin{equation}   \label{prefactor}
  A        : =          \begin{cases}  
  \sum_{k=1}^n       \frac{|\mu (s_k)|}{2\pi}   
   \frac{\left(\det \hess f(m_0) \right)^{\frac 12} }
  {\left| \det \hess f(s_k) \right|^{\frac 12}},      &      \text{ if }     f(m_0)< f(m_1)   ,      \\
     \sum_{k=1}^n       \frac{|\mu (s_k)|}{2\pi}   
   \frac{\left(\det \hess f(m_0) \right)^{\frac 12} + \left(\det \hess f(m_1) \right)^{\frac 12}}
  {\left| \det \hess f(s_k) \right|^{\frac 12}},   &    \text{ if }     f(m_0)= f(m_1)         .     \end{cases}
  \end{equation}
\noindent
Our second main theorem is the following.
\begin{theorem}\label{spectralgapF}
Assume~\ref{H2} and take $\ve_0>0$ as in Theorem~\ref{main1F}. Let $A,E$ be given respectively by~\eqref{exprate},~\eqref{prefactor} and let, for $\ve\in (0, \ve_0)$,  
 $\lambda(\ve)$  be the smallest non-zero eigenvalue  of $H_{\ve}$.  Then the 
error term $\mc R(\ve)$, defined for $\ve\in (0, \ve_0)$ by 
\begin{equation*}  \label{asymptoticlambda}      
 \lambda(\ve )         =          
  \ve    A   e^{-\frac{E}{\ve}}    \left( 1 + \mc R(\ve)\right), 
         \end{equation*}
satisfies the following: there exists a constant
$C>0$ such that  $| \mc R(\ve) |  \leq C\sqrt \ve$ for every $\ve\in (0, \ve_0)$.  
\end{theorem}

\begin{remark} 
Stronger smoothness properties of $f$ ($f\in C^4(\mb R^d)$ should suffice) may lead to the improved bound $\mc R_\ve = O(\ve)$. 
A possible proof may be obtained using the underlying Witten complex structure as 
explained in the author's PhD thesis~\cite{thesis}. There it is shown that $f\in C^\infty(\mb R^d)$ implies
that  $\mc R_\ve$ admits full asymptotic expansions in powers of $\ve$. But the
proof is substantially more involved, since it requires a construction and detailed analysis of
discrete WKB expansions on the level of $1$-forms. 
\end{remark}

\noindent
As anticipated in the introduction, our main results can be easily translated into results 
on spectral properties of the class of metastable discrete diffusions with
generator~\eqref{introL},~\eqref{introRates}. Since this might be a particularly interesting application of our results, we shall spell out precisely their
consequences from the stochastic point of view.

\

\subsection{ Results on the diffusion operator $L_{\ve}$}   

\

\noindent
Given a function $f:\mb R^d\to \mb R$ and a parameter $\ve>0$, we consider the weight functions
\begin{equation*}
 \rho_\ve(x)   =     e^{-\tfrac{f(x)}{\ve}}    \   \       \text{ and }   \    \      r_\ve(x,x')  =      \tfrac 1\ve  e^{- \tfrac{f(x') -f(x)}{2\ve } }          ,    \  \  \   \forall x, x'\in \mb R^d      .
\end{equation*}
Note that $\rho_\ve$ and  $r_\ve$  are related by the identity 
\begin{equation}\label{detailedbalance}      \rho_\ve(x)    r_\ve(x,x')      =      \rho_\ve(x')    r_\ve(x',x)         ,    \   \ \    \forall    \ve>0 \text{ and }   x,x'\in \mb R^d       . 
\end{equation}
We work now in the weighted Hilbert space $\ell^2(\rho_\ve)$
obtained as subspace
of  $\mb R^{\ve \mb Z^d}$ 
by introducing the weighted 
scalar product
\begin{equation*}     \lp \psi, \psi'\rp_{\ell^2(\rho_\ve)}         =   \ve^d    \sum_{x\in \ve \mb Z^d }    \psi(x) \, \psi'(x) \, \rho_\ve(x)    ,    
\end{equation*}  
 and the corresponding induced norm $\|\cdot\|_{\ell^2(\rho_\ve)}$.
We shall denote by $L_{\ve}$ the Laplacian of the weighted graph $\ve \mb Z^d$, whose vertices are weighted by $\rho_\ve$
and whose edges (determined by $\mc N$) are weighted by $\rho_\ve r_\ve$. More precisely we define 
$L_{\ve}:  {\textrm Dom} (L_{\ve})   \to    \ell^2(\rho_\ve)     $ by setting 
\[     {\textrm Dom} (L_{ \ve})     =    \left\{   \psi \in \ell^2(\rho_\ve):   \sum_{v\in \mc N}   r_\ve\left(x, x+  \ve v\right)     
         \left[       \psi(x+ \ve v  )      -    \psi(x )     \right]   \in \ell^2(\rho_\ve)          \right\}            ,        \]
and, for each $x\in \ve \mb Z^d$,
\begin{equation*}    
 L_{\ve}   \psi    (x)          =             \sum_{v\in \mc N}   r_\ve\left(x, x+  \ve v\right)     
         \left[       \psi(x+ \ve v  )      -    \psi(x )    \right]       ,   \     \     \  \   \      \forall     \psi \in {\textrm Dom} (L_{\ve})      .       
     \end{equation*}
This provides a Hilbert space realization of the formal operator~\eqref{introL},\eqref{introRates}.

\begin{proposition}\label{GST}
For each $\ve>0$ the operators $- \ve L_\ve$ and $H_\ve$ are unitarily equivalent. 
\end{proposition}
\begin{proof}
Let $\ve>0$. We consider the unitary operator 
\[\Phi_\ve:  \ell^2(\rho_\ve)   \to \ell^2(\ve \mb Z^d)     ,   \  \       \    \  \          \Phi_\ve [\psi]  (x) =  \sqrt{\rho_\ve}(x)  \psi (x)       .        \]
Then a direct computation shows that 
\begin{equation}\label{conjugation}
   H_{\ve} \psi             =         -  \ve \Phi_\ve \left[  L_{\ve}   \Phi^{-1}_\ve[\psi]  \right]         ,    \   \   \   \        \forall \psi\in {\textrm Dom(V_\ve)}, 
\end{equation}
and that $\Phi_\ve[ {\textrm Dom}(L_{\ve}) ] =   {\textrm Dom(V_\ve)}  $.
\end{proof}

\noindent
From the unitarily equivalence it follows that $L_{\ve}$ is not only symmetric 
and nonnegative (this can be checked by summation by parts and using the detailed balance 
condition~\eqref{detailedbalance}), but also selfadjoint.  
We remark also that $C_c(\ve \mb Z^d)$, which is a core for $H_{\ve}$ and is invariant under $\Phi_\ve$, is also a core of $L_{\ve}$.   
\

\noindent
Combining Proposition~\ref{GST} with Theorem~\ref{main1F} and Theorem~\ref{spectralgapF}
yields then the following result.

\begin{corollary}\label{MainCorollary}
 Assume~\ref{H1F} and denote by $N_0\in \mb N_0$ the number of local minima of $f$. There exist 
constants $\ve_0 \in (0, 1)$, $C>0$ such that
for each $\ve\in (0, \ve_0]$ the following properties hold true.
\begin{itemize}
\item[(i)]        $  \spec_{\textrm{ess}} (-L_{\ve})        \subset [\ve^{-1}C, \infty) $ and $    |  \spec_{\textrm{disc}} (-L_{\ve})   \cap [0, C]   |        \leq N_0$.
\item[(ii)]   $-L_\ve$ admits at least $N_0$ eigenvalues counting multiplicity. In the nontrivial case that 
$N_0\neq 0$, the $N_0$-th eigenvalue $\lambda_{N_0}(\ve)$ (according to increasing order and counting multiplicity) satisfies the bounds
\[0\leq \lambda_{N_0}(\ve)   \leq    e^{-\frac{C}{\ve}}.  \]
\end{itemize}
Moreover,  assuming in addition~\ref{H2},
and taking $A,E$ as in~\eqref{exprate},~\eqref{prefactor}, the
error term $\mc R(\ve)$, defined for $\ve\in (0, \ve_0)$ by 
\begin{equation}  \label{asymptoticlambda2}      
 \lambda_2(\ve )         =          
     A   e^{-\frac{E}{\ve}}    \left( 1 + \mc R(\ve)\right), 
         \end{equation}
         satisfies the following: there exists a constant
$C>0$ such that  $| \mc R(\ve) |  \leq C\sqrt \ve$ for every $\ve\in (0, \ve_0)$. 
\end{corollary}
\noindent
We stress that (i) implies a quantitative scale separation between the $N_0$ slow modes, corresponding to the metastable tunneling times, and all the other modes, 
corresponding to fast relaxations to local equilibria. In principle it is also possible to refine the analysis of the fast modes revealing the full hierarchy of scales governing the dynamics 
in the small $\ve$ regime, see~\cite{DGMariani} for the continuous space setting and a $\Gamma$-convergence formulation.

  As already mentioned, the rigorous derivation of an Eyring-Kramers formula of type~\eqref{asymptoticlambda2} in the setting of discrete metastable diffusions had already been derived by a different approach
based on capacity estimates~\cite{BEGKdiscI, BdH}.  
Compared to these previous results the formula given in~\eqref{asymptoticlambda2} differs in two aspects: 
\begin{itemize}
\item[1)] The estimate on the error term
 $ \mc R(\ve)$ is improved by our approach, since in~\cite[Theorem 10.9 and 10.10]{BdH}, under the same regularity assumptions as considered here ($f\in C^3(\mb R^d)$) a logarithmic correction appears. 
More precisely our result improves the error estimate from  $\mc R(\ve)= \mc O(\sqrt{\ve [\log{1/\ve}]^3})$ to $\mc R(\ve)= \mc O(\sqrt{\ve})$. 
\item[2)] The prefactor $A$ given in~\eqref{asymptoticlambda2} differs from the one given in~\cite{BEGKdiscI, BdH}. 
This is due to our slightly different choice of jump rates, compare~\eqref{introRates} with~\cite[(10.1.2.), p. 248]{BdH}. Indeed 
 it is clear that the prefactor is sensible to the particular choice of jump rates among the infinitely many 
possible jump rates satisfying the detailed balance condition with respect to the Boltzmann weight $ e^{-f/\ve}$. This sensitivity of the prefactor is opposed to the robustness of the exponential rate $E$, 
which is universal as can be seen e.g. via a Large Deviations analysis. We remark that, while the rates chosen in~\cite{BdH} correspond to a Metropolis algorithm, our choice~\eqref{introRates} corresponds, 
in the context of the Statistical Mechanics models mentioned above,
to a heat bath algorithm. This is a very natural choice and is considered for example in~\cite{LMT}. As observed in the introduction, it is the choice which in first order approximation gives the same prefactor
as the continuous space model~\eqref{introHcont}. Furthermore,~\cite{BEGKdiscI,BdH} concerns discrete time processes, which means that the rates are normalized and thus bounded over $\mb R^d$.
Our setting includes also the case of possibly unbounded rates which requires some additional technical work for the analysis outside compact sets. 
\end{itemize}

\

\section{General tools for a semiclassical analysis on the lattice}\label{SectionTools}

\noindent
This section is devoted to some preliminary tools for a semiclassical analysis on the lattice. 

Subsection~\ref{SSIMS} concerns a discrete IMS localization formula, see~\cite[Lemma 11.3]{Teschl}
or~\cite[Theorem 3.2]{CFKS}, where also an explanation of the name can be found, for the standard continuous space setting and~\cite{KleinRosenbergerII}. The IMS formula is a simple observation based on a computation of commutators.
It will be used repeteadly for decomposing the quadratic form induced by a Schrödinger operator into localized parts. 

Subsection~\ref{SSHarmonic} provides estimates on the first two eigenvalues of the discrete semiclassical Harmonic oscillator. These estimates follow from more general results proven in~\cite{KleinRosenbergerII}. 
Nevertheless we shall include a relatively short and completely selfcontained proof, which focuses on the estimates needed to prove the separation between exponentially small eigenvalues of $H_\ve$ and the rest of its spectrum, as provided by Theorem~\ref{main1F}. The proof is based on a microlocalization which permits to separate high and low frequency actions of the operator. 

Subsection~\ref{sectiondiscreteLaplace} provides sharp asymptotic results for Laplace-type sums. These are instrumental in almost all the computations necessary for deriving the Kramers formula for the eigenvalue splitting
and for tunneling calculations in general. 
Our proofs are again based on Fourier analysis. In particular, following~\cite{thesis},  we shall use the Poisson summation formula: shifting a function by an integer vector and summing over all shifts produces the same periodization as taking the Fourier series of the Fourier transform. Compared to~\cite{thesis}, where it is shown how to get complete asymptotic expansions in the smooth setting, here we shall relax the regularity assumptions on the phase function to cover the applications we have in mind.

\subsection{The discrete IMS formula}\label{SSIMS}

\

\noindent
We say that the set $\{\chi_j\}_{j\in J}$ is a smooth quadratic partition of unity of $\mb R^d$
if $J$ is a finite set, $\chi_{j}\in C^\infty(\mb R^d)$ for every $j\in J$ and $\sum_{j\in J} \chi^2_{j} \equiv 1$. 

\begin{proposition}\label{propdiscreteIMSF}  
There exists a constant $C>0$ such that for every 
$\ve>0$, every $\psi\in \ell^2(\ve \mb Z^d)$ and every smooth quadratic partition of unity 
$\{\chi_j\}_{j\in J}$ it holds 
\[   \Big\|\Delta_\ve \psi      -           \sum_{j\in J}    \chi_j   \, \Delta_\ve \left( \chi_j  \,  \psi \right) \Big\|_{\ell^2(\ve\mb Z^d)}     \leq    C  \sup_{x,j} |\hess \chi_{j}(x)|     \| \psi\|_{\ell^2(\ve\mb Z^d)}    .   \]
\end{proposition}
\begin{proof}
We have
\begin{gather*}      \Delta_\ve \psi   -    \sum_j       \chi_j   \, \Delta_\ve \left( \chi_j \,  \psi \right)   (x)      =         
\tfrac {1}{\ve^2}   \sum_{v\in \mc N}  \left[ 1 - \sum_j \chi_j(x)  \chi_j(x+\ve v)  \right]   \psi(x+\ve v)       ,  
\end{gather*}
thus
\begin{gather}      \Big\|\Delta_\ve \psi   -    \sum_j    \chi_j   \, \Delta_\ve \left( \chi_j \,  \psi \right) \Big\|_{\ell^2(\ve\mb Z^d)}     \ 
\leq    \nonumber \\     
\tfrac {1}{\ve^2}   \sum_{v\in \mc N}  \sup_{x\in \mb R^d} \left| 1 - \sum_j \chi_j(x)  \chi_j(x+\ve v)  \right|   \|   \psi(\cdot+\ve v)  \|_{\ell^2(\ve\mb Z^d)}         .  
 \label{montiF} 
 \end{gather}
Differentiating the relation $\sum_j\chi^2_j \equiv 1$ yields  $\sum_j \chi_j \nabla \chi_j  \cdot v \equiv 0 $ for every $v$ and therefore, 
by Taylor expansion, for every $x\in \mb R^d$ and $v\in \mc N$,
\begin{gather} \label{treviF}
\left| 1 - \sum_j \chi_j(x)  \chi_j(x+\ve v)  \right|   \leq         \tfrac{\ve^2}{2}    \sup_{y\in \mb R^d}  \sum_j |\chi_j(y)|   \, |  \hess \chi_j(y) v\cdot v |      . 
\end{gather}
The claim follows now from~\eqref{montiF} and~\eqref{treviF} by noting that the assumption 
$\sum_j\chi^2_j \equiv 1$ also implies $\sup_{j, x} |\chi_{j}(x)|\leq 1$,   
that  $ \|   \psi(\cdot+\ve v)  \|_{\ell^2(\ve\mb Z^d)}    =    \| \psi  \|_{\ell^2(\ve\mb Z^d)} $ for every $v$
and recalling that $\mc N$ is bounded.
 \end{proof}

\subsection{Estimates on the discrete semiclassical Harmonic oscillator}\label{SSHarmonic}

\

\noindent
We provide lower bounds for the 
first and the second eigenvalue of the semiclassical discrete Harmonic oscillator.
\begin{proposition}\label{propdiscreteharmonicoscillatorF} 
For every $x\in \mb R^d$ let $U(x) =    \lp x-\bar x, M (x-\bar x) \rp $, where
$\bar x\in \mb R^d$ and $M$ is a symmetric $d\times d$ real matrix with strictly positive eigenvalues
denoted by $\kappa_1, \dots, \kappa_d$. Moreover let 
$\lambda_0 =  \sum_j \sqrt{\kappa_j}  $ and 
$\lambda_	1 =  \sum_j \sqrt{\kappa_j} +    2 \min_j \sqrt{\kappa_j}   $.   
Then 
there exist  for every $\ve>0$ a function $\Psi_{\ve}\in \ell^2(\ve \mb Z^d)$ and constants $\ve_0, C>0$ such that
for every $\ve\in (0, \ve_0]$ and
$\psi\in C_c(\ve\mb Z^d)$ the following hold:
\begin{itemize}
\item[(i)]   $ \lp \left(-\ve^2\Delta_\ve       + U  \right) \psi, \psi   \rp_{\ell^2(\ve\mb Z^d)  }      \geq        \ve  \left( \lambda_0   -      C\ve^{\frac 15}    \right)   \|\psi\|^2_{\ell^2(\ve\mb Z^d) }  . $
\item[(ii)]   $ \lp \left(-\ve^2\Delta_\ve      + U  \right) \psi, \psi   \rp_{\ell^2(\ve\mb Z^d)  }      \geq          \ve  \left( \lambda_1   -      C\ve^{\frac 15}    \right)    \|\psi\|^2_{\ell^2(\ve\mb Z^d) }          -      \lp  \psi, \Psi_\ve   \rp^2_{\ell^2(\ve\mb Z^d)  }               .     $
\end{itemize}
\end{proposition}
\noindent
The proof is by localization around low frequencies in Fourier space and comparison with the corresponding continuous Harmonic oscillator on $\mb R^d$,
whose first and second eigenvalue are given respectively by 
$\ve  \lambda $ and $\ve   \lambda_1  $. At low frequencies, discrete
and continuous Harmonic oscillators are close, while the high frequencies do not contribute to the bottom of the spectrum.

 In the proof we shall use the following notation: for $\ve>0$ and
  $\psi \in \ell^1(\ve \mb Z^d)$ we define 
 \[    \hat \psi(\xi)      :=      (2\pi)^{-\frac{d}{2} }
\sum_{x\in \ve \mb Z^d}  \psi(x)   \,     e^{- \frac{ \im  x\cdot \xi}{\ve}}        \    \  \    \text{ for }     \xi\in \mb R^d         ,     \]
and for $\ve>0$ and
  $\phi \in L^1(\mb R^d)$ we define 
\begin{equation*}  \label{JF}     \check\phi (\xi)       :=           (2\pi)^{-\frac{d}{2} }  \int_{\mb R^d}      \phi (x)     \,     e^{ \frac{\im   x\cdot \xi}{\ve}}  \, dx    
   \    \  \    \text{ for }     \xi\in \mb R^d          .   
\end{equation*}
Then by Parseval's theorem
\begin{equation}\label{Parseval}   \|\psi\|_{\ell^2(\ve\mb Z^d)}      =  
  \|\hat{\psi}\|_{L^2([-\pi, \pi]^d)}   \   \   \    \  \        \forall \psi \in \ell^2(\ve \mb Z^d)      ,      
\end{equation}
and by Plancherel's theorem
\begin{equation}\label{Plancherel}     \|\phi\|_{L^2(\mb R^d)}     =   
 \|\check\phi\|_{L^2(\mb R^d)}   \   \   \      \  \       \forall \phi \in L^1(\mb R^d)\cap L^2(\mb R^d)     .  
\end{equation}
We recall also the inversion theorem for the Fourier transform and Fourier series, which in our notation reads as follows. Let 
$\phi\in S(\mb R^d)$, the Schwartz space on $\mb R^d$ and let $\tilde \phi (x) = \check\phi (-x)$
for every $x\in \mb R^d$. Then 
\begin{equation} \label{FinversionT}     \phi(\xi)     =     \check{ \tilde \phi} (\xi)          \  \  \    \  \    \forall \xi \in    \mb R^d.     \end{equation}    
Moreover, for every $\phi\in C^\infty(\mb R^d)$ with $\supp( \phi) \subset (-\pi, \pi)^d$ it holds $\check \phi\in \ell^1(\ve \mb Z^d)$
and 
\begin{equation} \label{FSinversionT}    
 \phi(\xi)     =     \hat{ \check \phi} (\xi)          \  \  \    \  \    \forall \xi \in    [-\pi, \pi]^d.     \end{equation}

\begin{proof}[Proof of Proposition~\ref {propdiscreteharmonicoscillatorF} ]
Let $\ve\in (0,1]$, $\psi\in C_c(\ve\mb Z^d)$ and 
let $\varphi := \left(-\ve^2\Delta_\ve     + U  \right) \psi$. Then 
$\phi\in C_c(\ve\mb Z^d)$ and  $\hat \varphi      =    \left( W      -    A_\ve \right) \hat\psi$, 
where $W:\mb R^d\to \mb R$ is a multiplication operator given by 
\begin{equation*}\label{symbolW}     W(\xi)   :=          4  \sum_{j=1}^d   \sin^2 \left( \frac{  \xi_j}{2} \right)    ,   
\end{equation*}
and $A_\ve$ is a second order differential operator given by 
                  \[   A_\ve    :=        \sum_{j,k=1}^d M_{j,k} \,  \left(   \ve^2 \partial_j\partial_k    + \ve \,  2 \bar x_k \im \partial_j  -   \bar x_k \bar x_j  \right)         .   \]
It follows then by Parseval's theorem~\eqref{Parseval} that 
\begin{equation} \label{FourierrepresentationscalarF}
\lp \left(-\ve^2\Delta_\ve   \    + U  \right) \psi, \psi   \rp_{\ell^2(\ve\mb Z^d)  }         =  
\lp \left( W      -    A_\ve \right) \hat\psi, \hat\psi   \rp_{L^2([-\pi, \pi]^d)  }      .  \end{equation}
We now consider a  cut-off function $\theta \in C^\infty(\mb R^d;[0,1])$ which equals $1$ on $\{\xi: |\xi| \leq 1\}$ and vanishes on $\{\xi: |\xi| \geq 2\}$.
For $j=1, \dots, N$ we define with $s=\tfrac 25$ the $\ve$-dependent smooth quadratic partition of unity 
$\{  \theta_{0,\ve},   \theta_{1,\ve}\}$ by setting
\[     \theta_{0,\ve} (\xi)      :=          \theta \left(\ve^{-s}   (\xi)\right)     \  \   \    \   \  
 \theta_{1, \ve} (\xi)       :=       \sqrt{ 1 -  \theta^2_{0,\ve} (\xi) }     
      .       \]
Moreover we denote by $W_0$ the leading term in the $\xi$-expansion of the function $W$ around the origin, i.e. 
\[          W_0(\xi)     :=        \tfrac12 \hess W(0)\,  \xi\cdot \xi      =        |\xi|^2     \   \ \     \ \   
 \forall \xi\in \mb R^d           .    \]
A simple rearrangement of terms gives 
\begin{gather}
\lp \left( W      -    A_\ve \right) \hat\psi, \hat\psi   \rp_{L^2([-\pi, \pi]^d)  }     =     
\lp \left( W_0      -   A_\ve \right) \theta_{0,\ve}\hat\psi,   \theta_{0,\ve}    \hat\psi   \rp_{L^2([-\pi, \pi]^d)  }   \  +   
\nonumber \\    
\lp \left( W -     A_\ve \right)  \,  \theta_{1,\ve}\hat\psi,   \theta_{1,\ve}    \hat\psi   \rp_{L^2([-\pi, \pi]^d)  }      +    
\mc E_1(\ve)    +   \mc E_2  (\ve)       ,   
\label{sumFourierImsF}
\end{gather}
where the localization errors $\mc E_1(\ve)  , \mc E_2 (\ve)$ are given by
\[    \mc E_1(\ve)     :=       \lp \left( W      -  W_0 \right) \theta_{0,\ve}\hat\psi,   \theta_{0,\ve}    \hat\psi   \rp_{L^2([-\pi, \pi]^d)  }      ,      
   \]
    \begin{equation*}  
            \mc E_2(\ve)     :=     
   -  \sum_{j=0}^1 \lp   \left(  \theta_{j,\ve } A_\ve - A_\ve \theta_{j,\ve } \right)  \hat \psi,    \theta_{j,\ve}\hat \psi    \rp_{L^2([-\pi, \pi]^d)  }   . 
  \end{equation*}
 The four terms in the right hand side of~\eqref{sumFourierImsF} are analyzed separately in the following.

\

{    \it 1) Analysis of the first term in the right hand side of~\eqref{sumFourierImsF}.}

\, 

\noindent
Using that  $\supp \theta_{0,\ve}\subset (-\pi, \pi)^d$ for $\ve\in (0,1]$, 
Plancherel's theorem~\eqref{Plancherel} and that the smallest eigenvalue of
the Harmonic Oscillator    $- \ve^2 \Delta  + U$ on $\mb R ^d$ is $\lambda_0$, gives 
\begin{gather}  
\lp \left( W_0      -    A_\ve  \right) \theta_{0,\ve}\hat\psi,   \theta_{0,\ve}    \hat\psi   \rp_{L^2([-\pi, \pi]^d)  }     =   
\lp \left( W_0      -    A_\ve  \right) \theta_{0,\ve}\hat\psi,   \theta_{0,\ve}    \hat\psi   \rp_{L^2(\mb R^d)  }         =   \nonumber \\
\lp \left( -    \ve^2 \Delta  + U    \right) \widecheck{\theta_{0,\ve}\hat\psi},  \widecheck{ \theta_{0,\ve}    \hat\psi }  \rp_{L^2(\mb R^d)  } 
  \geq  \ve   \lambda_0     \Big\| \widecheck{\theta_{0,\ve}\hat\psi} \Big \|^2_{L^2(\mb R^d) }   =     \nonumber \\
  \ve  \lambda_0  \| \theta_{0,\ve}\hat\psi  \|^2_{L^2(\mb R^d) }     =    \ve  \lambda_0     \| \theta_{0,\ve}\hat\psi \|^2_{L^2([-\pi, \pi]^d) }     \   \  \    \   \   \forall \ve\in (0,1]     .    \label{part1samF}       
\end{gather}
Moreover, considering for $\ve>0$ the ground state  
\[    g_\ve (x)     :  =          \frac{e^{-\frac{\lp x, \sqrt M x\rp}{2\ve}}}{\big\|e^{-\frac{\lp x, \sqrt M x\rp}{2\ve}}\big\|_{L^2(\mb R^d)}}          \                 ,    \]
an analogous computation gives for $\ve\in (0,1]$ the estimate
\begin{gather*}
\lp \left( W_0      -    A_\ve \right) \theta_{0,\ve}\hat\psi,   \theta_{0,\ve}    \hat\psi   \rp_{L^2([-\pi, \pi]^d)  }   
   =  
\lp \left( -    \ve^2 \Delta  + U    \right) \widecheck{\theta_{0,\ve}\hat\psi},  \widecheck{ \theta_{0,\ve}    \hat\psi }  \rp_{L^2(\mb R^d)  } 
  \geq     \nonumber   \\
  \ve   \lambda_0      \lp \widecheck{\theta_{0,\ve}\hat\psi}, g_\ve  \rp^2_{L^2(\mb R^d) }   + 
   \ve\lambda_1   
   \left(   \|\widecheck{\theta_{0,\ve}\hat\psi}\|_{L^2(\mb R^d) } ^2    -    \lp \widecheck{\theta_{0,\ve}\hat\psi}, g_\ve  \rp^2_{L^2(\mb R^d) }       \right)       =      \nonumber \\ 
    \ve \lambda_1   
      \|\widecheck{\theta_{0,\ve}\hat\psi}\|_{L^2(\mb R^d) } ^2      -    
       2 \ve  \min_j\sqrt{ \kappa_j}      \lp  \widecheck{\theta_{0,\ve}\hat\psi}, g_\ve  \rp^2_{L^2(\mb R^d) } = \\
    \ve  \lambda_1   \| \theta_{0,\ve}\hat\psi \|^2_{L^2([-\pi, \pi]^d) }  
     -  2 \ve  \min_j\sqrt{ \kappa_j}      \lp  \widecheck{\theta_{0,\ve}\hat\psi}, \check{\tilde {g}}_\ve  \rp^2_{L^2(\mb R^d) },      
 \end{gather*}
 where for the last equality the Fourier inversion theorem~\eqref{FinversionT} is used for $g_\ve$.  
Moreover using~\eqref{Plancherel},~\eqref{FSinversionT} and~\eqref{Parseval} we get
\begin{gather*}
  \lp  \widecheck{\theta_{0,\ve}\hat\psi}, \check{\tilde {g}}_\ve  \rp_{L^2(\mb R^d) }    =    
\lp  \theta_{0,\ve}\hat\psi, \tilde {g}_\ve  \rp_{L^2(\mb R^d) }    = 
\lp \hat\psi,  \theta_{0,\ve}  \tilde {g}_\ve  \rp_{L^2([-\pi, \pi]^d)  }    =  \\
\lp \hat\psi,  \widehat{\widecheck{\theta_{0,\ve}  \tilde {g}}}_\ve  \rp_{L^2([-\pi, \pi]^d)  }   =        
\lp  \psi,  \check\theta_{0,\ve} \tilde {g}_\ve  \rp_{L^2([-\pi, \pi]^d)  } =
  \lp  \psi   ,  \widecheck{\theta_{0,\ve} \tilde {g} }_\ve \rp_{\ell^2(\ve\mb Z^d) }      .       
\end{gather*}
Thus, setting for shortness 
 \begin{equation*}  \label{cJF?}
       \Phi_\ve(\xi)        :=     \sqrt{2 \ve}  \min_j  \kappa^{\frac 14}_j  \tilde {g}_\ve (\xi)
   \    \  \    \text{ for }     \xi\in \mb R^d          ,   
\end{equation*}    
we can conclude that 
  \begin{gather}
\lp \left( W_0      -    A_\ve  \right) \theta_{0,\ve}\hat\psi,   \theta_{0,\ve}  
  \hat\psi   \rp_{L^2([-\pi, \pi]^d)  }      \geq     \label{part2samF}     \\
        \ve \lambda_1  
        \| \theta_{0,\ve}\hat\psi \|^2_{L^2([-\pi, \pi]^d) }        -         
      \lp  \psi   ,  \widecheck{\theta_{0,\ve} \Phi_\ve } \rp^2_{\ell^2(\ve\mb Z^d) }      \   \      \  \    \   \forall   \ve\in (0,1]    .    \nonumber 
\end{gather}

\

{     \it 2) Analysis of the second term in the right hand side of~\eqref{sumFourierImsF}.}

\, 

\noindent
Using the inequality $\sin\tfrac t2   \geq     \tfrac t4$ for $t\in [0,\pi]$ gives
\begin{equation} \label{LBpincio}
    W(\xi)       =    
   4  \sum_{j=1}^d   \sin^2 \left( \frac{  \xi_j}{2} \right)      \geq        \frac 14 |\xi|^2   
         \   \  \    \   \   \forall \xi\in [-\pi, \pi]^d     .         \end{equation}
Since      
$\supp \theta_{1,\ve} \subset \{  \xi\in \mb R^d:  |\xi| \geq 2 \ve^{\tfrac 25}  \}$ for $\ve\in (0,1]$,
the bound~\eqref{LBpincio} implies
 \begin{equation} \label{LBawayFF} 
\lp  W    \theta_{1,\ve}\hat\psi,   \theta_{1,\ve}    \hat\psi   \rp_{L^2([-\pi, \pi]^d)  }         \geq   
 \ve^{\frac 45}  \|\theta_{1,\ve}\hat\psi\|^2_{L^2([-\pi, \pi]^d)  }       \  \    \   \  \   \forall \ve\in (0,1].
\end{equation}
Moreover, since $\hat\psi$ is periodic and $\theta_{1,\ve}$ equals $1$ around the boundary of $[-\pi, \pi]^d$
for $\ve\in (0,1]$, integration by parts gives
\begin{equation} \label{intbypartstorusF}
\lp -    A_\ve   \,  \theta_{1,\ve}\hat\psi,   \theta_{1,\ve}    \hat\psi   \rp_{L^2([-\pi, \pi]^d)  }      \geq  0
          \  \   \  \   \   \forall \xi\in (0, 1]         .         \end{equation}
In particular, it follows from~\eqref{LBawayFF} and~\eqref{intbypartstorusF}
that there exists an $\ve'_0\in (0,1]$ such that for all $\ve\in (0,\ve'_0]$
 \begin{equation}   \label{number2localizationF}
 \lp \left( W -     A_\ve \right)  \,  \theta_{1,\ve}\hat\psi,   \theta_{1,\ve}    \hat\psi   \rp_{L^2([-\pi, \pi]^d)  }  
        \geq   
  \ve \lambda_1     
     \|\theta_{1,\ve}\hat\psi\|^2_{L^2([-\pi, \pi]^d)  }     . 
 \end{equation}

\

{    \it 3) Analysis of the localization error $\mc E_1$. }

\,

\noindent
Since by Taylor expansion there exists a $C'>0$ such that $ |W(\xi)-W_0(\xi)| \leq C' |\xi|^3$
for $|\xi| \leq 2$, one gets
\begin{gather*}     
|\mc E_1 (\ve)|       =    \left| \lp \left( W      -  W_0 \right) \theta_{0,\ve}\hat\psi,   \theta_{0,\ve}    \hat\psi   \rp_{L^2([-\pi, \pi]^d)  }   \right|   \leq    \\
    \sup_{|\xi| < 2\ve^{\frac 25}} |W(\xi)-W_0(\xi)|   \,  \| \theta_{0,\ve}    \hat\psi   \|^2_{L^2([-\pi, \pi]^d)}       \leq    
     8 C'   \, \ve^{\frac 65}  \,      \| \theta_{0,\ve}    \hat\psi   \|^2_{L^2([-\pi, \pi]^d)}      \  \   \   \  \   
       \forall    \ve\in (0,1]   .  
 \end{gather*}
 In particular, since
  $  \| \theta_{0,\ve}    \hat\psi   \|^2_{L^2([-\pi, \pi]^d)} \leq   \|   \hat\psi   \|^2_{L^2([-\pi, \pi]^d)}
 =    \| \psi   \|^2_{\ell^2(\ve \mb Z^d)}$, it follows that 
 \begin{equation}       \label{localizationerror1F}
    \mc E_1(\ve)       \geq     -    8C'  \, \ve^{\frac 65}  \,  
 \| \psi   \|^2_{\ell^2(\ve \mb Z^d)}   \   \  \    \   \    \forall \ve \in (0,1]      .     
 \end{equation}

\

{     \it 4)  Analysis of the localization error $\mc E_2$. }

\,

\noindent
A straightforward computation (see  also \cite[Lemma 11.3]{Teschl}) gives the IMS localization
formula
\[      A_\ve \hat\psi     - \sum_{j=0}^1 \theta_{j, \ve}  A_\ve   (\theta_{j, \ve} \hat \psi)    =         \ve^2 \sum_{j=0}^1   
\lp \nabla \theta_{j, \ve} ,  M \nabla \theta_{j, \ve}   \rp \hat \psi  (\xi )        \   \text{ on }  \mb R^d      .      \]
Thus there exists a constant $C''>0$ such that 
 \begin{gather*}  
         | \mc E_2 (\ve)|      \leq   
      \ve^2 \sum_{j=0}^1  \sup_{\xi\in \mb R^d} 
       |\lp \nabla \theta_{j, \ve} (\xi),  M \nabla \theta_{j, \ve}   (\xi )\rp|   \ 
      \|\hat \psi    \|_{L^2([-\pi, \pi]^d)  }       \leq     \\
           C'' \, \ve^{2-2s}  \,   \|\hat \psi    \|_{L^2([-\pi, \pi]^d)  }        \  \    \  \  \  \forall \ve\in (0,1]     .
  \end{gather*}
 Recalling that $s=\tfrac 25$ and the Parseval theorem~\eqref{Parseval} we conclude that
  \begin{equation}\label{localizationerror2F}
   \mc E_2  (\ve)       \geq       -    C'' \ve^{\frac 65}  \,  
 \| \psi   \|^2_{\ell^2(\ve \mb Z^d)}   \  \  \     \  \        \forall \ve \in (0,1]        .     
 \end{equation}

\

{   \it Final step.}

\, 

\noindent
Statement (i) in the theorem follows by putting together~\eqref{FourierrepresentationscalarF}, \eqref{sumFourierImsF},
 \eqref{part1samF}, \eqref{number2localizationF}, \eqref{localizationerror1F} and \eqref{localizationerror2F}, chosing
 $C= 8C' + C''$ and $\ve_0= \ve_0'$ and observing that
 \[      \| \theta_{0,\ve}    \hat\psi   \|^2_{L^2([-\pi, \pi]^d}      +    
    \| \theta_{1,\ve}    \hat\psi   \|^2_{L^2([-\pi, \pi]^d}    =    
       \|    \hat\psi   \|^2_{L^2([-\pi, \pi]^d}  =        \| \psi   \|^2_{\ell^2(\ve \mb Z^d)}   .   \]
Statement (ii) follows similarly, but using~\eqref{part2samF} instead of~\eqref{part1samF}
 and chosing 
$\Psi_\ve =   \widecheck{\theta_{0,\ve} \Phi_\ve }_{\Big |\ve\mb Z^d} $. 
\end{proof}

\subsection{Laplace asymptotics on $\ve\mb Z^d$}\label{sectiondiscreteLaplace}

\

\noindent
Given $x_0\in \mb R^d$ and $\delta>0$ we denote by 
$B_\delta(x_0) = \{   x\in \mb R^d :    |x-x_0| <  \delta \}$ 
the open ball of radius $\delta$ around $x_0$ and, for each $\ve > 0 $, by 
$B_\delta^\ve(x_0) = B_\delta(x_0)  \cap \ve \mb Z^d$
its  intersection with $\ve \mb Z^d$ and by 
$[B^\ve_\delta( x_0)]^c  =   \ve \mb Z^d\setminus B^\ve_\delta( x_0)$ the complementary of 
$B^\ve_\delta( x_0)$.

\begin{proposition}\label{propgaussianlaplace}
Let $q(x) = \tfrac 12  x \cdot Q x $, where $Q$ is a symmetric, positive definite  $d\times d$ matrix and let $x_0\in \mb R^d$ and $m \in \mb N_0$. Then there exists a $\gamma>0$ such that
for every $\ve\in (0,1]$
\begin{equation}\label{gaussasymptoticsFF}
\ve^{\tfrac d2}      \sum_{x\in \ve\mb Z^d}  |x-x_0|^{2m}     e^{ -\frac{ q(x-x_0)    }{\ve}}    =    
      \ve^{m} \int_{\mb R^d}  |x-x_0|^{2m}     e^{ - q(x-x_0)    }  \, dx 
      +  \mc O(e^{-\frac \gamma \ve})   . 
 \end{equation}
Moreover for every $\delta>0$ there exists $\gamma(\delta)>0$  such that
for every $\ve\in (0,1]$
\begin{equation}\label{discretetailFF}   
      \sum_{x\in B^\ve_\delta( x_0)}   |x-x_0|^{2m}    e^{ -\frac{ q(x-x_0)    }{\ve}}   =
         \sum_{x\in \ve\mb Z^d}    |x-x_0|^{2m}    e^{ -\frac{ q(x-x_0)    }{\ve}}  
            \left( 1 +    \mc O(e^{-\frac {\gamma(\delta)}{ \ve}}) \right)   .
\end{equation}       
\end{proposition}
\begin{remark}
The Gaussian integrals appearing on the right hand side of~\eqref{gaussasymptoticsFF} can be computed explicitly.  We shall use in the sequel the explicit value only for $m=0$, in which case~\eqref{gaussasymptoticsFF} becomes  
\begin{equation*}\label{gaussasymptoticsFcorF}
\ve^{\tfrac d2}      \sum_{x\in \ve\mb Z^d}      e^{ -\frac{ q(x-x_0)    }{\ve}}    =   
     \sqrt{\tfrac{(2\pi)^{d}}{\det Q} }    +   
   \mc O(e^{-\frac \gamma \ve})       .
   \end{equation*}
We shall also use the following estimate for odd moments: 
\begin{equation}\label{gaussasymptoticsFcorF}
\ve^{\tfrac d2}      \sum_{x\in \ve\mb Z^d}  |x-x_0|^m     e^{ -\frac{ q(x-x_0)    }{\ve}}    =   
\mc O(\ve^{\frac m2})    \        \text{ for  }    m= 1, 3, \dots  \, ,     
\end{equation}
The latter follows from Proposition~\ref{propgaussianlaplace} and the Cauchy-Schwarz inequality
\begin{gather*}    \left|  \ve^{\tfrac d2}      \sum_{x\in \ve\mb Z^d}  |x-x_0|^m      e^{ -\frac{ q(x-x_0)    }{\ve}}  \right|   \leq   \\    
\left( \ve^{\tfrac d2}      \sum_{x\in \ve\mb Z^d}  |x-x_0|^{2 m}  e^{ -\frac{ q(x-x_0)    }{\ve}}    \right)^{\frac 12}    \left(    
     \ve^{\tfrac d2}      \sum_{x\in \ve\mb Z^d}      e^{ -\frac{ q(x-x_0)    }{\ve}}  \right)^{\frac 12}. 
     \end{gather*}
 \end{remark}

\begin{proof}[Proof of Proposition~\ref{propgaussianlaplace}] 
\

\noindent
The function $x\mapsto u(x) := |x|^{2m} e^{-q(x-x_0)}$ is in the Schwartz space $\mc S(\mb R^d)$
and its Fourier transform $\hat u (x) :=\int_{\mb R^d} u(y)\ 
e^{-2\pi\im x\cdot y} \ dy$ satisfies the Poisson summation formula (see e.g. \cite[Corollary 2.6, p. 252]{Steinweiss})
 \begin{equation*} 
\sum_{x\in\mb Z^d} u ( x)  =   \sum_{x\in\mb Z^d} 
\hat u(x)  . 
\end{equation*}
It follows that
 \begin{gather*}
  \ve^{\tfrac d2}      \sum_{x\in \ve\mb Z^d}  |x-x_0|^{2m}     e^{ -\frac{ q(x-x_0)    }{\ve}}    =     
   \ve^{\tfrac d2 +  m}      \sum_{x\in   \mb Z^d}  |\sqrt \ve ( x -    \tfrac{x_0}{ \ve})|^{2m}      \,  e^{  q(\sqrt \ve ( x -    \tfrac{x_0}{ \ve})   )  }    =        \\
  =      \ve^{\tfrac d2 + m}      \sum_{x\in   \mb Z^d}  u(\sqrt \ve ( x -    \tfrac{x_0}{ \ve}) )    =     
          \ve^{m}         \sum_{x\in   \mb Z^d}  e^{ -    \tfrac{2\pi i x \cdot x_0}{ \ve} }     \hat u(  \tfrac {x}{ \sqrt {\ve}} )       =   
         \ve^{m}      \int_{\mb R^d}   u(x)  \, dx   +  R_\ve     ,
 \end{gather*}
 with 
 \[    R_\ve   :=    \ve^{m}         \sum_{x\in   \mb Z^d \setminus\{0\}}  e^{ -    \tfrac{2\pi i x \cdot x_0}{ \ve} }     \hat u(  \tfrac {x}{ \sqrt {\ve}} )        .     \]
Since $\hat u$ is a linear combination of derivatives of Gaussian functions, there exist constants
$C, \gamma>0$ such that
\begin{equation*} \label{FourierdecayFF}    
       |\hat u(x) |      \leq    C e^{-2\gamma |x|^2}   
     \  \  \   \  \ 
  \forall    x\in \mb R^d          .
   \end{equation*}
It follows that for every $\ve\in (0,1]$
\[     |R_\ve|   \leq   C    \ve^{m}   \sum_{x\in   \mb Z^d \setminus\{0\}}       e^{-\frac{2\gamma |x|^2}{\ve}}     =     C    \ve^{m}   e^{-\frac{\gamma}{\ve}}     \sum_{x\in   \mb Z^d \setminus\{0\}}       e^{-\frac{\gamma }{\ve}(2|x|^2 - 1)} 
\leq C'     e^{-\frac{\gamma}{\ve}}    ,   \]
with  $C'   :=    C \sum_{x\in   \mb Z^d \setminus\{0\}}       e^{-\gamma (2|x|^2 - 1)}    $ 
which concludes the proof of~\eqref{gaussasymptoticsFF}.

\

\noindent
In order to prove~\eqref{discretetailFF}, fix $\delta>0$ and note that, due to the positive definiteness of $Q$, there exists a constant $C>0$ such that   $q(x) >    C\delta^2$ for every 
$x\in [B^\ve_\delta( x_0)]^c $. Thus, for $\ve\in (0, 1]$,
 \begin{gather}\sum_{x\in [B^\ve_\delta( x_0)]^c}     e^{ -\frac{ q(x-x_0)    }{\ve}}         =    
 e^{ -\frac{ C\delta^2 }{\ve}}      \sum_{x\in [B^\ve_\delta( x_0)]^c}     e^{ -\frac{ q(x-x_0)  - C\delta^2   }{\ve}}      \leq     \nonumber \\
 \leq   \ve^{-d} e^{ -\frac{ C\delta^2 }{\ve}}           e^{  C\delta^2}     \ve^d   \sum_{x\in [B^\ve_\delta( x_0)]^c}     e^{ - q(x-x_0) }     
 \leq      \ve^{-d} e^{ -\frac{ C\delta^2 }{\ve}}  K    ,   
 \label{ghyuF} 
   \end{gather} 
with  $K =   e^{ C\delta^2} \left( \int_{\mb R^d}   
 e^{ - q(x-x_0) }    dx  +1 \right)$. To see the last inequality one can 
 use e.g. the Poisson summation formula for 
 $  \ve^d   \sum_{x\in \ve \mb Z^d}     e^{ - q(x-x_0) }       $. 
  From~\eqref{ghyuF}, chosing 
   $\gamma >0$ sufficiently small and $C'>0$ sufficiently large we obtain 
 \[  \sum_{x\in [B^\ve_\delta( x_0)]^c}     e^{ -\frac{ q(x-x_0)    }{\ve}}     \leq    C'    e^{-\frac{\gamma}{\ve}} .   \]
 The estimate~\eqref{discretetailFF} for $m=0$ follows then 
 using~\eqref{gaussasymptoticsFF} with $m=0$. The case of positive 
 $m$ can be proven in the same way. 
 \end{proof}
\noindent
The following proposition concerns more general, not necessarily quadratic phase functions. 
\begin{proposition}\label{GenPropositionLaplaceF}
Let $x_0\in \mb R^d$, $\delta>0$, $k\in \{3,4\}$ and  $\vp \in C^k(\overline{B_\delta(x_0)})$ s.t. 
\begin{equation}\label{laplaceassumptionvpF}
\vp(x_0)=0,    \    
 \hess \vp (x_0)>0  \text{ and 
 }   \vp(x)>0  \text{ for every }  
x\in    \overline {B_\delta(x_0)}  .   
\end{equation}
Moreover let $m \in \mb N_0$. Then  for $\ve \in (0,1]$ it holds
\begin{gather}
\ve^{\tfrac d2}      \sum_{x\in B^\ve_\delta( x_0) }  |x-x_0|^{2m}     e^{ -\frac{ \vp(x)    }{\ve}}    =   \nonumber  \\ 
\ve^{m} \int_{\mb R^d}  |x-x_0|^{2m}     e^{ - q(x-x_0)  }  \, dx 
   \left( 1    +    \mc O(\ve^{\tfrac{k-2}{2}}) \right)    ,    \label{discLaplacestatementF}
\end{gather}       
where $q(x)  =     \tfrac 12 \hess \vp(x_0) x\cdot x$ for all $x\in \mb R^d$.  
\end{proposition}
\begin{remark} Under the stronger regularity assumption  $\vp\in C^\infty(B_\delta(x_0))$ one can show that the error term 
in~\eqref{discLaplacestatementF} admits a complete asymptotic expansion in powers of $\ve$, see~\cite[Appendix C]{thesis} for details. 
\end{remark}
\begin{proof}
We reduce the problem to the quadratic case of Proposition~\ref{propgaussianlaplace}.
For   
 $x\in \overline{B_\delta(x_0) } $ let for short
 $r(x)=   \vp(x)   -   q(x-x_0)$ and note that 
 there exist $\alpha, \delta'>0$ such that 
 $\tilde q(x) :=  \alpha  |x|^2 $   satisfies
 \begin{equation}    \label{LBtildephiF}
   \tilde \vp(x)    :=    q(x-x_0)   -      |r(x)|   \geq   \tilde q(x-x_0)    
 \     \   \ \   \       \forall x\in \overline{ B_{\delta'}(x_0)}    
,      
\end{equation}
and also  $ \vp(x)   \geq   \tilde q(x-x_0) $
for all  $x\in \overline{ B_\delta(x_0)}$. Indeed, 
the assumption $\vp\in C^3(\overline{B_\delta(x_0)})$ implies the
existence of a constant
$C>0$ such that  $|r(x)| \leq C |x-x_0|^3$ for all $x\in 
\overline{B_\delta(x_0)}$. It follows that, 
denoting by $\lambda>0$ the smallest eigenvalue of $\hess \vp(x_0)$
and taking e.g. $\delta' =   \tfrac{\lambda}{4C}$ and $\alpha' =   \tfrac{\lambda}{4}    $, 
\[     \tilde \vp (x)      \geq  
    (\tfrac{\lambda}{2}    -     C\delta')   |x-x_0|^2  \geq  
    \tfrac{\lambda}{4}     |x-x_0|^2      \  \   \  \  \       
    \forall  x\in \overline{B_{\delta'}(x_0) }   .  \]
Note that, a fortiori, also $ \vp(x)\geq \tilde q(x)$ for every $x
\in B_{\delta'}(x_0)$. Moreover, since $\tfrac{\vp(x)}{|x|^2}$ is continuous
and stricly positive on the compact set 
$\overline{B_{\delta}(x_0)} \setminus B_{\delta'}(x_0)$ we can take e.g.
\[    \alpha     =     \min\{\alpha', \inf_{x\in \overline{B_{\delta}(x_0)} \setminus B_{\delta'}(x_0)}\tfrac{\vp(x)}{|x|^2}\}      .     \]

\

\noindent
It will be enough to prove~\eqref{discLaplacestatementF}
with the sum  on the left hand side restricted
to  $B_{\delta'}^\ve(x_0)$, since by
Proposition~\ref{propgaussianlaplace} there exists a 
$\gamma>0$ s.t. for $\ve\in (0,1]$
\begin{gather*}
    \sum_{x\in B^\ve_\delta( x_0) \setminus B^\ve_{\delta'}( x_0) }  |x-x_0|^{2m}     e^{ -\frac{ \vp(x)    }{\ve}}    \leq   
     \sum_{x\in [B^\ve_{\delta'}( x_0)]^c } 
 |x-x_0|^{2m}     e^{ -\frac{ \tilde q(x-x_0)    }{\ve}}     =   
 \mc O(e^{-\gamma/\ve})     .    
\end{gather*} 
We shall  consider the decomposition  
\begin{equation}\label{decompositionLaplaceF}
\ve^{\tfrac d2}      \sum_{x\in B^\ve_{\delta'}( x_0) }  |x-x_0|^{2m}      e^{ -\frac{ \vp(x)    }{\ve}}   =     I_0(\ve)    + I_1(\ve) + I_3(\ve)   , 
\end{equation}
with, setting for short  $u_\ve(x) = |x-x_0|^{2m}      e^{ -\frac{ q(x)    }{\ve}} $, 
\begin{equation*}
 I_0(\ve)   =      \ve^{\tfrac d2}      \sum_{x\in B^\ve_{\delta'}( x_0) }
  u_\ve(x)    ,     \   \      I_1(\ve)   =      \ve^{\tfrac d2}      \sum_{x\in B^\ve_{\delta'}( x_0) }  
     \ve^{-1} r( x) u_\ve(x)    
\end{equation*}
and
\begin{equation*}
 I_2(\ve)   =      \ve^{\tfrac d2}      \sum_{x\in B^\ve_{\delta'}( x_0) } 
      \left(   e^{-\frac{r(x)}{\ve}}  -1  - \ve^{-1} r(x)  \right) 
      u_\ve(x)     . 
\end{equation*}
It follows from Proposition~\ref{propgaussianlaplace} that there exists a 
$\gamma>0$ s.t. for $\ve\in (0,1]$
\begin{equation}\label{LLL0}   I_0(\ve)    =        \ve^{m} \int_{\mb R^d}  |x-x_0|^{2m}     e^{ - q(x-x_0)    }  \, dx 
      +  \mc O(e^{-\frac \gamma \ve})  .   \end{equation}
Morever, using $|r(x)| \leq C |x-x_0|^3$ for all 
$x\in  \overline{B_\delta(x_0)}$  and~\eqref{LBtildephiF} gives
 \begin{gather}
 |I_2(\ve) |   \leq      \tfrac 12  \ve^{\tfrac d2-2}     
  \sum_{x\in B^\ve_{\delta'}( x_0) }     |r(x)|^2   
   e^{\frac{|r(x)|}{\ve}}  u_\ve(x)      \leq   \nonumber  \\
    \tfrac C2  \ve^{\tfrac d2-2}     
  \sum_{x\in B^\ve_{\delta'}( x_0) }  |x-x_0|^{2m+6}  
     e^{ -\frac{ \tilde q(x-x_0)    }{\ve}}       =     \mc O(\ve^{m+1})     , \label{LLL2} 
\end{gather}
with the last estimate being a consequence of Proposition~\ref{propgaussianlaplace}. Finally, in order to 
analyze the term $I_1(\ve)$, we consider first the case $k=3$.
We then have by~\eqref{gaussasymptoticsFcorF} 
\begin{gather*}
 |I_1(\ve)|   \leq      \ve^{\tfrac d2-1}    C  \sum_{x\in B^\ve_{\delta'}( x_0) }  
|x-x_0|^{2m +3}   e^{ -\frac{ q(x-x_0)    }{\ve}}     = 
 \mc O(\ve^{m+\frac 12})     , 
 \end{gather*}
which together with~\eqref{decompositionLaplaceF}, \eqref{LLL0}
and~\eqref{LLL2} finishes the proof for $k=3$.
For the case $k=4$ we write
 $r(x) =    t_3(x) +
 \rho(x)$, 
where $t_3: \overline{B_{\delta'}( x_0)} \to \mb R$ is the cubic term in the Taylor expansion of $\vp$ around $x_0$, thus satisfying $t_3(x_0+x) = t_3(x_0-x)$, 
and 
$\rho: \overline{B_{\delta'}( x_0)} \to \mb R$ 
 satisfies $|\rho(x)| \leq C' |x-x_0|^4$ for some $C'>0$.
We then have
\begin{gather*}
 I_1(\ve)   =      \ve^{\tfrac d2-1}      \sum_{x\in B^\ve_{\delta'}( x_0) }  
     \left( t_3(x) +
 \rho(x) \right)   u_\ve(x)   =  
  \ve^{\tfrac d2-1}      \sum_{x\in B^\ve_{\delta'}( x_0) }  
     \rho(x)   u_\ve(x) , 
     \end{gather*}
     and therefore  by Proposition~\ref{propgaussianlaplace}
\begin{gather*}    
| I_1(\ve) | \leq   \ve^{\tfrac d2-1}  C' \sum_{x\in B^\ve_{\delta'}( x_0) }  |x-x_0|^{2m+4}   
  e^{ -\frac{ q(x-x_0)    }{\ve}}    =   \mc O(\ve^{m+1}) , 
    \end{gather*}
which finishes the proof in the case $k=4$.
\end{proof}

\section{Proof of Theorem~\ref{main1F}}\label{SectionMainProof1}

\noindent
Recall the definition of $V_\ve$ given in~\eqref{mainV}. To  prove Theorem~\ref{main1F} we shall reduce to suitable localized problems and 
then exploit basic pointwise estimates on $V_\ve$ as stated in the following two complementary lemmata. The first one gives a uniform strictly positive lower bound 
on $V_\ve$ away from critical points. The second one concerns the local behavior of $V_\ve$ around critical points. Note that these bounds are almost immediate to obtain, even under weaker assumptions, if instead of $V_\ve$ one considers the corresponding continuous space potential $ \tfrac 14 |\nabla f|^2 - 
 \tfrac{\ve}{2} \Delta f$ appearing in~\eqref{introHcont}. The discrete case follows from straightforward Taylor expansions and elementary estimates. We shall give the details of the arguments 
 at the end of this section for completeness.
\begin{lemma} \label{lowerboundV*F}
Assume $f\in C^2(\mb R^d)$ and that  $\hess f$ is bounded on $\mb R^d$.
Let $S\subset \mb R^d$ and  $a>0$ such that $|\nabla f(x)|> a$
 for every $x\in S$.
 Then there exist constants $\ve_0,C>0$ such that
 \begin{equation*}\label{lowerboundUF}      
V_\ve(x)        \geq          C       \   \  \     \   \  \    \forall x\in S   \   \text{ and  }   \   \forall \ve \in (0, \ve_0]               .  
\end{equation*}
\end{lemma}
\begin{lemma}\label{localV*F}     
Assume $f\in C^3(\mb R^d)$. Let $z\in \mb R^d$ such that $\nabla f(z)=0$, $R>0$ and 
\[   U(x)    :=   \frac 14\lp \left[ \hess f (z)\right]^2  (x-z), (x-z)\rp   .    \] 
Then there exists a constant $C>0$ such that for all $x\in B_R(z)$ and $ \ve>0$
\begin{equation*}
|V_\ve(x)       -     U(x)   +   \frac{\ve}{2}\Delta f(z)  |   \leq   
 C  \left(  |x-z|^3  +    \ve \, |x-z|    +    \ve^2 \right).
 \end{equation*}
 \end{lemma}

\noindent
After these preliminary estimates on $V_\ve$ we turn to the proof of  Theorem~\ref{main1F}. 
We first show that 
 the essential spectrum of 
$H_{\ve}$ is bounded from below by a constant, as claimed in Theorem~\ref{main1F} (i).  
\begin{proposition}[Localization of the essential spectrum]  \label{PropessF}
Under Assumption~\ref{H1F} 
 there exist constants $\ve_0, C>0$ such that
\begin{equation*}\label{LocEss2}            \spec_{\textrm{ess}} (H_{\ve})      \subset [C, \infty)     \    \   \          \forall  \ve \in (0,\ve_0]      .   
\end{equation*}
\end{proposition}
\begin{remark}
The proof given below shows that the claim of Proposition~\ref{PropessF} still holds without assuming that the critical points of $f$ are nondegenerate. Also the regularity assumption on $f$
can be relaxed by assuming $f\in C^2(\mb R^d)$ instead of $f\in C^3(\mb R^d)$. 
\end{remark}
\begin{proof}
Let $\chi := \alpha \mbf 1_{K}$, where $\mbf 1_{K}$ is the indicator function
 of a bounded set $K\subset \mb R^d$ and $\alpha\in \mb R$.
Then $\chi$, seen as a multiplication operator in $\ell^2(\ve\mb Z^d)$, 
is  of finite rank 
(in particular compact) for every $\ve>0$. It follows from Weyl's
theorem that for fixed $\ve>0$,
\begin{equation}\label{step1harnmonicapprox}
 \inf \spec_{\textrm{ess}}
\big( H_{\ve}\big)  = 
\inf  \spec_{\textrm{ess}}
\big( H_{\ve} +  \chi  \big)  . 
\end{equation}
 Moreover  
\begin{gather*}\inf \ \spec_{\textrm{ess}}
\big( H_{\ve}  +  \chi  \big)  \geq 
\inf\spec 
\big( H_{\ve}  +  \chi  \big) 
=     \\  
 \inf_{\substack{\psi\in {\textrm Dom}(V_\ve) \\
\psi\neq 0}} 
\frac{\lp (H_{\ve}+\chi)\psi,
\psi\rp_{\ell^2(\ve\mb Z^d)}}{\lp\psi,\psi\rp_{\ell^2(\ve\mb Z^d)}} 
\geq   \inf_{\substack{\psi\in {\textrm Dom}(V_\ve) \\
\psi\neq 0}} 
\frac{\lp (V_{\ve}+\chi)\psi,
\psi\rp_{\ell^2(\ve\mb Z^d)}}{\lp\psi,\psi\rp_{\ell^2(\ve\mb Z^d)}}   .               
\end{gather*}
The claim follows by  chosing $\alpha$ and $K$ large enough so that for some constants
$\ve_0,C>0$ the inequality $V_\ve(x)+\chi(x)\geq C$
 holds for every
$x\in\mb R^d$ and $\ve\in (0,\ve_0]$.
To see that this choice is possible recall the uniform bound $V_\ve \geq -2d$
and note that 
by Assumption~\ref{H1F} (i) there exist $a>0, R>0$ such that $|\nabla f(x)| >a$ for $|x| > R$. 
It follows then by Lemma~\ref{lowerboundV*F} that for suitable $C, \ve_0>0$
it holds $V_\ve(x) \geq C$ for $|x|>R$ and $\ve\in (0, \ve_0)$.  
\end{proof}
\noindent
The next proposition provides the crucial estimate for the proof of statement (ii) in Theorem~\ref{main1F}.

\begin{proposition}\label{ProproughF}
Assume~\ref{H1F} and denote by $N_0\in \mb N_0$ the number of local minima of $f$. Then 
there exist constants $\ve_0, C>0$ and, for every $\ve>0$, functions $\Psi_{1,\ve},\dots, \Psi_{N_0, \ve}\in \ell^2(\ve\mb Z^d)$
such that   for every   $ \psi\in {\textrm Dom}(V_\ve) $  it holds 
\begin{equation}\label{boundmodulodimF}
\lp  H_{\ve} \psi, \psi    \rp_{\ell^2(\ve\mb Z^d)}      \geq   
 C  \ve \  \|\psi\|^2_{\ell^2(\ve\mb Z^d)}       -      \sum_{k=1}^{N_0}   \lp    \psi   , \Psi_{k,\ve}   \rp^2_{\ell^2(\ve\mb Z^d)}        \   \  \  
\forall  \ve\in (0,\ve_0] .
\end{equation}
\end{proposition}
\noindent
Statement (ii) in Theorem~\ref{main1F} is then a simple consequence of the Max-Min principle (see e.g. \cite[Theorem 11.7]{Helffer2013}):
\begin{corollary} \label{cormain2}
Assume~\ref{H1F} and denote by $N_0\in \mb N_0$ the number of local minima of $f$. Then 
there exist constants $\ve_0, C>0$ such that
 \[   |  \spec_{\textrm{disc}} (H_{\ve})   \cap [0, C\ve]   |        \leq N_0         \   \  \   \forall \ve\in (0, \ve_0] .   \]
\end{corollary}
\begin{proof}[Proof of Corollary~\ref{cormain2}]
By Proposition~\ref{PropessF} and Proposition~\ref{ProproughF} we can find 
 $\ve_0,C>0$ 
 such that 
\begin{equation}\label{LocEss}      
      \spec_{\textrm{ess}} (H_{\ve})      \subset [C\ve_0, \infty)     \    \   \          \forall  \ve \in (0,\ve_0]         
\end{equation}
and such that~\eqref{boundmodulodimF} holds.     
If for every $\ve\in (0, \ve_0)$ it happens that 
\[     |  \spec_{\textrm{disc}} (H_{\ve})   \cap [0,  \tfrac{C \ve}{2} ]   |        \leq N_0      ,  \]
the claim is proven. 
 Thus, we only have to check the case in which there exists $\ve_*\in (0,\ve_0)$ such that 
\begin{equation}\label{auxpropcorollminimax}    |  \spec_{\textrm{disc}} (H_{\ve_*})   \cap [0,  \tfrac{C \ve_*}{2} ]   |        > N_0      .  
\end{equation}
But this case is impossible.    
 Indeed~\eqref{auxpropcorollminimax}
implies that  there exist at least $N_0+1$ distinct eigenvalues of 
$H_{\ve_*}$ in $[0,\tfrac{C \ve_*}{2}]$ and thus in particular the $N_0+1$-th eigenvalue $\lambda_{N_0+1}(\ve_*)$(in increasing order and counting multiplicity) exists and satisfies 
\begin{equation}\label{6783one} 
\lambda_{N_0+1}(\ve_*) \leq \tfrac{C \ve_*}{2} .   
\end{equation}
In particular $\lambda_{N_0+1}(\ve_*) \leq \tfrac{C \ve_0}{2} $ and therefore, by~\eqref{LocEss}, $\lambda_{N_0+1}(\ve_*)$ is smaller than the bottom of the essential spectrum. 
 From this, the Max-Min principle and~\eqref{boundmodulodimF} 
it follows that
\begin{equation}\label{6783two}          \lambda_{N_0+1}(\ve_*)     \geq   \inf_{\psi}   \lp  H_{\ve_*} \psi, \psi    \rp_{\ell^2(\ve\mb Z^d)}     \geq     C\ve_*    ,    \end{equation}
where the infimum is taken over all normalized $\psi \in \mc V_\ve^{\perp} \cap {\textrm Dom(V_\ve)}$, 
with $\mc V_\ve$  being the linear span of the set $\{\Psi_{1,\ve},\dots, \Psi_{N_0, \ve}\}\subset \ell^2(\ve\mb Z^d)$ appearing in~\eqref{boundmodulodimF}.  
But~\eqref{6783one} and~\eqref{6783two} are in contradiction.
\end{proof}

\begin{proof}[Proof of~Proposition~\ref{ProproughF}]
We label by $z_1,\dots, z_N$ the critical points of $f$, the ordering being chosen such that
$z_1, \dots, z_{N_0}$ are the local minima. 
Then we take a function $\chi\in C^\infty(\mb R^d;[0,1])$ which equals $1$ on $\{x: |x| \leq 1\}$ and vanishes on $\{x: |x| \geq 2\}$. We shall consider a smooth quadratic partition of unity by defining with $s=\tfrac 25$
\[     \chi_{j,\ve} (x)       :=       \chi \left(\ve^{-s}   (x-z_j )\right)    ,   \   \   \  
 \chi_{0, \ve} (x)       :=       \left( 1 -   \sum_j  \chi^2_{j,\ve} (x) \right)^{\frac 12}     
           \]
for $j=1, \dots, N$ and $\ve\in (0, \overline\ve]$, where $\overline \ve \in (0,1]$ is sufficiently small so that $ \chi_{0, \ve} \in C^\infty(\mb R^d) $. 
We set moreover for $x\in \mb R^d$ and $j=1, \dots, N$
\begin{gather*} 
  U_{j} (x)      :=         \tfrac 14 \lp \left[ \hess f (z_j)\right]^2  (x-z_j), (x-z_j)\rp          .
\end{gather*}
Let $\psi\in {\textrm Dom}(V_\ve)$. It follows from $\sum_{j=0}^N   \chi^2_{ j,\ve}  \equiv 1$ that we can write
\begin{gather}
\lp   H_\ve \psi, \psi    \rp_{\ell^2(\ve\mb Z^d)}      =               
   \sum_{j=1}^N \lp   \left( -\ve^2 \Delta_\ve + U_{j}    -      \tfrac{\ve}{2} \Delta f (z_j)     \right)   \chi_{j, \ve }\psi,   \chi_{j,\ve } \psi    \rp_{\ell^2(\ve\mb Z^d)} 
   +   \nonumber \\           
     \lp   \left( -\ve^2 \Delta_\ve + V_\ve \right)  \chi_{0,\ve}\psi,    \chi_{ 0,\ve}   \psi    \rp_{\ell^2(\ve\mb Z^d)}   
      +             \mc E_1       +               \mc E_2     ,     \label{sumImsF}
\end{gather}
with the localization errors given by 
\begin{equation*}
 \mc E_1         =       \mc E_1(\ve)     :=     \sum_{j=1}^N   \lp   \left( V_\ve  - U_{j} +  \tfrac{\ve}{2} \Delta f (z_j)       \right)  \chi_{j,\ve}\psi,    \chi_{j,\ve}   \psi    \rp_{\ell^2(\ve\mb Z^d)}       ,
  \end{equation*}
  \begin{equation*}  
  \mc E_2         =          \mc E_2(\ve)   :=     
   -\ve^2 \sum_{j=0}^N \lp   \left(  \chi_{j,\ve }\Delta_\ve - \Delta_\ve \chi_{j,\ve } \right)  \psi,  \chi_{j,\ve }   \psi    \rp_{\ell^2(\ve\mb Z^d)}  . 
  \end{equation*}

The four terms in the right hand side of~\eqref{sumImsF} are now analyzed separately. 
 
 \
 
 {\it 1) Analysis of the first term in the right hand side of~\eqref{sumImsF}. }
 
 \, 

\noindent
We apply Proposition~\ref{propdiscreteharmonicoscillatorF}: let $\kappa_1(z_j)\dots, \kappa_d(z_j)$ be the eigenvalues of $\tfrac 12\hess f(z_j)$, so that in particular $\tfrac 12 \Delta f(z_j)= \sum_{i}\kappa_i (z_j)$
and $\kappa^2_1(z_j)\dots, \kappa^2_d(z_j)$ are the eigenvalues of $\tfrac 14\left[\hess f(z_j) \right]^2$. 

\

\   \   \  \   \     {\it Case 1: $j=1, \dots, N_0$ (i.e. $z_j$ is a local minimum of $f$)} 

\

\noindent
In this case $\sum_i  \left( |\kappa_i(z_j)|   -   \kappa_i(z_j)\right) =0$ and according to Prop.~\ref{propdiscreteharmonicoscillatorF} (ii)
there exist  for every $\ve>0, j=1,\dots, N_0$ a function $\Phi_{j,\ve}\in \ell^2(\ve \mb Z^d)$ and constants $\ve'_{0}, C'>0$ such that
for every $\ve\in (0, \ve'_{0}]$, $j=1, \dots, N_0$ and
$\psi\in C_c(\ve\mb Z^d)$ 
\begin{gather}
\lp \left(-\ve^2\Delta_\ve      +  U_j  - \tfrac{\ve}{2} \Delta f(z_j)     \right) \chi_{j, \ve}\psi, \chi_{j, \ve} \psi   \rp_{\ell^2(\ve\mb Z^d)  }    \geq      \nonumber \\       
C'   \ve     \| \chi_{j, \ve }\psi\|^2_{\ell^2(\ve\mb Z^d) }          -      \lp   \chi_{j, \ve }\psi , \Phi_{j,\ve}   \rp^2_{\ell^2(\ve\mb Z^d)  }   .      \label{minimumcaseF}    
\end{gather}

\

\   \   \  \   \     {\it Case 2: $j=N_0+1, \dots, N$ (i.e. $z_j$ is not a local minimum of $f$)} 

\

\noindent
In this case $\sum_i  \left( |\kappa_i(z_j)|   -   \kappa_i(z_j)\right) > 0$ and according to Prop.~\ref{propdiscreteharmonicoscillatorF} (i), possibly taking 
the constants $\ve'_{0}, C'>0$ smaller, the following holds:
for every $\ve\in (0, \ve'_{0}]$, $j=N_0+1, \dots, N$ and
$\psi\in C_c(\ve\mb Z^d)$ 
\begin{equation} \label{nonminimumcaseF}
\lp \left(-\ve^2\Delta_\ve      +  U_j  - \tfrac{\ve}{2} \Delta f(z_j)     \right) \chi_{j, \ve}\psi, 
\chi_{j, \ve} \psi   \rp_{\ell^2(\ve\mb Z^d)  }      \geq   
C'   \ve  \   \| \chi_{j, \ve }\psi\|^2_{\ell^2(\ve\mb Z^d) }          .    
\end{equation}

\

{      \it 2) Analysis of the second term in the right hand side of~\eqref{sumImsF}.}

\, 

\noindent
According to Lemma~\eqref{localV*F} there exist constants $r, C''>0$ and $\ve''_0\in (0,\overline\ve]$ such that
for every $j=1, \dots, N$
\begin{equation}\label{pescaF}
V_\ve (x)      \geq        C''  |x-z|^2    -    \tfrac {\ve}{ C''}      \  \ \   \ \    \forall x\in   B_r(z_j)    \text{ and }   \forall\ve\in (0,\ve''_0]     .  
\end{equation}
Moreover, according to Lemma~\eqref{lowerboundV*F}, possibly taking the constants $\ve_0'', C''>0$ smaller, it holds also
\begin{equation}\label{pesca2F}
V_\ve (x)     \   \geq   \    C''     \  \   \     \     \        \forall  x\in \mb R^d\setminus \bigcup_{j=1}^N B_r(z_j) 
  \text{ and }      \forall\ve\in (0,\ve''_0]         . 
\end{equation}
 Since      
$\supp \chi_{0,\ve} \subset \{  x\in \mb R^d:  |x-z_j| \geq 2 \ve^{\tfrac 25}  \text{ for all } j=1, \dots, N   \}$,
the lower bounds~\eqref{pescaF},~\eqref{pesca2F} imply (with possibly reducing further the constant
$\ve_0''>0$)  
 \begin{equation*} 
\lp  V_\ve    \chi_{0,\ve} \psi,   \chi_{0,\ve}   \psi   \rp_{\ell^2(\ve \mb Z^d)  }        \geq  
 C''\ve^{\frac 45}  \|\chi_{0,\ve}\psi\|^2_{\ell^2(\ve \mb Z^d)  }        \    \    \   \  \   \forall \ve\in (0,\ve_0'']
                .         \end{equation*}
Using $-\Delta_\ve \geq 0$  we conclude that  
 \begin{equation}   \label{pesca3F}
  \lp   \left( -\ve^2 \Delta_\ve + V_\ve \right)  \chi_{0,\ve}\psi,    \chi_{ 0,\ve}   \psi    \rp_{\ell^2(\ve\mb Z^d)} 
     \geq    
 C''\ve        \|\chi_{0,\ve}\psi\|^2_{\ell^2(\ve \mb Z^d)  }         \  \    \  \  \  \forall \ve\in (0,\ve''_0]
                 .         \end{equation}

\

{    \it 3) Analysis of the localization error $\mc E_1$.}

\,

\noindent
Let $R_{j,\ve}(x) :=   V_\ve(x) - U_j(x) +\tfrac \ve 2 \Delta f(z_j)$. By Lemma~\ref{localV*F} 
there exist constants $C''', \ve_0'''>0$ such that  
\[    \sup_{x: |x-z_j| \leq 2\ve^{\frac 25} } |R_{j,\ve}(x)|    \leq         C'''   \ve^{\frac 65}  \ \  
\   \  \      \forall \ve\in (0, \ve_0''']    \text{ and }     \forall j=1, \dots, N      .     \]
Thus, since     
$\supp \chi_{j,\ve} \subset \{  x\in \mb R^d:  |x-z_j| \leq 2 \ve^{\tfrac 25}\}$
  for all $j=1, \dots, N$,
\begin{gather*}     
|\mc E_1 (\ve)|       =        \left| \sum_{j=1}^N \lp R_{j,\ve} \chi_{j,\ve}\psi,  \chi_{j,\ve}   \psi   \rp_{\ell^2(\ve \mb Z^d)  }   \right|     \leq    
     C'''    \, \ve^{\frac 65}  \,    \sum_{j=1}^N   \| \chi_{j,\ve} 
   \psi   \|^2_{\ell^2(\ve \mb Z^d)  }      \  \   \  \   \     \forall    \ve\in (0,\ve'''_0]  ,  
 \end{gather*}
 and we conclude that
 \begin{equation}  \label{2localizationerror1FF}
    \mc E_1(\ve)       \geq       -   C''' \ve^{\frac 65}  \,   
 \|  \psi   \|^2_{\ell^2(\ve \mb Z^d)  }  \   \  \  \   \    \forall \ve \in (0,\ve_0''']         .     
 \end{equation}

\

{    \it 4)   Analysis of the localization error $\mc E_2$.}

\,

\noindent
Using $\sum_{j=0}^N \chi_{j, \ve}^2 \equiv 1$ gives 
 \begin{equation*}  
    \mc E_2(\ve)    =   
   -\ve^2   \lp   \left(  \Delta_\ve -  \sum_{j=0}^N  \chi_{j,\ve }\Delta_\ve \chi_{j,\ve } \right)  \psi,   \psi    \rp_{\ell^2(\ve\mb Z^d)}     . 
  \end{equation*}
Since there is a constant $K>0$ such that $\sup_{x\in \mb R^d} |\hess \chi_{j,\ve}(x)| \leq K \ve^{-2s}$
for every $\ve\in (0, \overline \ve]$ and $j=0, \dots, N$, 
it follows from Lemma~\ref{propdiscreteIMSF} that there exists a constant $C''''>0$  such that  
 \begin{gather*}  
           | \mc E_2 (\ve)|   
         \leq       C''''\, \ve^{2-2s}  \,   \| \psi    \|_{\ell^2(\ve\mb Z^d)  }  
        =      C''''\, \ve^{\frac 65}  \,   \| \psi    \|_{\ell^2(\ve\mb Z^d)  }      \  \   \  \  \  \forall \ve \in (0, \overline\ve]    .
  \end{gather*} 
  In particular we shall use that 
  \begin{equation}  \label{appdiscreteIMSF}
           \mc E_2 (\ve)  
         \geq      -
         C''''\, \ve^{\frac 65}  \,   \| \psi    \|_{\ell^2(\ve\mb Z^d)  }      \   \   \  \  \   \forall \ve \in (0, \overline\ve]    .
  \end{equation}

\

{    \it Final step.}

\,

\noindent
Taking $\Psi_{j,\ve}:= \chi_{j,\ve} \Phi_{j, \ve}$, $\tilde\ve_0 :=\min \{\ve_0', \ve_0'', \ve_0'''\}$
and $\tilde C:=\min\{C', C''\}$ gives, according to~\eqref{sumImsF},~\eqref{minimumcaseF},~\eqref{nonminimumcaseF},~\eqref{pesca3F},~\eqref{2localizationerror1FF},~\eqref{appdiscreteIMSF} 
the lower bound
\begin{gather*}\lp   H_\ve \psi, \psi    \rp_{\ell^2(\ve\mb Z^d)}      \geq    \\
\left( \tilde C   \ve     -     (C'''' + C'''') \, \ve^{\frac 65}  \right) \   \| \psi\|^2_{\ell^2(\ve\mb Z^d) }     
      -      \sum_{j=1}^{N_0} \lp  \psi , \Psi_{j,\ve}   \rp^2_{\ell^2(\ve\mb Z^d)  }     \   \ \  \  \  \forall \ve\in (0, \tilde\ve_0]       ,    
      \end{gather*}
 which implies the desired estimate~\eqref{boundmodulodimF}, by taking $C =\tilde C/2$ and 
 a sufficiently small $\ve_0\in (0,\tilde \ve_0)$ . 
   \end{proof}
\noindent
It remains to show part (iii) of Theorem~\ref{main1F} to complete the proof. 
In order to do so, we can assume $N_0\neq 0$, since otherwise there is nothing to prove. 
By the Max-Min principle~\cite[Theorem 11.7 and Proposition 11.9]{Helffer2013}
together with the bound on the essential spectrum given by Proposition~\eqref{PropessF}
it is sufficient to show that for each $\ve>0$ there exist $N_0$ orthonormal
functions in the domain ${\textrm Dom}(V_\ve)$  of $H_\ve$  
such that the quadratic form associated with $H_\ve$ is exponentially small for each of these functions. We shall now exhibit such a family of orthonormal functions.

\

\noindent
Let $\{z_1,\dots, z_{N_0}\}$ be the set of local minima of $f$. We fix $\delta>0$ such that $B_{3\delta}(z_k)\cap B_{3\delta}(z_j)$ is the empty set 
 for $k\neq j$ and such that  $f> f(z_k)$ on $B_{3\delta}(z_k)\setminus\{z_k\}$. 
 Moreover we fix for each $k=1, \dots , N_0$ a cutoff function $\chi_k\in C^\infty(\mb R^d;[0,1])$, 
satisfying  $\chi\equiv 1$ on $B_\delta(z_k)$
$\chi \equiv  0$ on  $\mb R^d\setminus B_{2\delta}(z_k)$.
We consider then for each $\ve>0$ and for each 
$k=1, \dots, N_0$ the functions $\psi_{k,\ve}:\mb R^d \to \mb R$ given by 
\begin{equation}\label{quasimodesforroughF}
\psi_{k,\ve}   (x)    =       \frac{\chi_k( x)   e^{- f(x)/(2\ve)}  }{\|\chi_k  
 e^{- f/(2\ve)}\|_{\ell^2(\ve\mb Z^d)} }   . 
\end{equation}
Then for each $\ve>0$ the (restrictions to $\ve\mb Z^d$ of the)  functions 
$\psi_{1,\ve}, \dots, \psi_{N_0, \ve}$ are in the domain of $H_\ve$ and orthonormal 
in $\ell^2(\ve\mb Z^d)$. Moreover the following proposition shows that
the quadratic form associated with $H_\ve$ is exponentially small for each of these functions
and thus concludes the proof of Theorem~\ref{main1F}.
 
\begin{proposition}Assume~\ref{H1F} and that the set $\{z_1,\dots, z_{N_0}\}$ of local minima of $f$ is not empty. Then there exist $C, \ve_0>0$ such that for each $\ve\in (0,\ve_0]$ the functions $\psi_{1,\ve}, \dots, \psi_{N_0, \ve}$ defined in~\eqref{quasimodesforroughF} satisfy 
the estimate
\begin{equation*}
     \lp H_\ve \psi_{k,\ve}, \psi_{k, \ve}\rp_{\ell^2(\ve\mb Z^d)}    \leq   e^{- C/\ve }    . 
\end{equation*}
\end{proposition}

\begin{proof} Fix $k=1, \dots, N_0$. Then, applying Proposition~\ref{GenPropositionLaplaceF}
with $\vp=f- f(z_k)$, $k=3$ and $m=0$, gives for a suitable constant $K>0$ and for every 
$\ve\in (0,1]$
\begin{equation} \label{minimaxnormF}
\|\chi_k  
 e^{- f/(2\ve)}\|^2_{\ell^2(\ve\mb Z^d)}   \geq   \ve^d  e^{-f(z_k)/\ve} \sum_{x\in B_\delta^\ve(z_k)} 
  e^{- (f(x)-f(z_k)/(2\ve)}  \geq    K \ve^{\frac d2} e^{-f(z_k)/\ve} . 
\end{equation}
Further, using~\eqref{Witten*} and the notation $F_\ve(x,v) =
\frac 12 [ f(x) + f(x+\ve v) ] $, 
\begin{equation}\label{minmaxWittF}
     \lp H_\ve (\chi_k  e^{- f/(2\ve)}) , \chi_k   e^{- f/(2\ve) } \rp_{\ell^2(\ve\mb Z^d)} = 
     \ve^2 \|e^{-F_\ve/(2\ve)}\nabla_\ve \chi_{k} \|^2_{\ell^2(\ve \mb Z^d; \mb R^\mc N)}    . 
\end{equation}
We take $\ve'_0\in (0,1]$ small enough such that
for all $\ve\in (0, \ve_0]$ it holds
\begin{equation*}
 \nabla_\ve \chi_{k} (x, v) = 0        \   \  \  \    \    \forall x\in  B_{\delta/2}(z_k)  
 \text{ and }  \forall v\in \mc N. 
\end{equation*}
Moreover we take $\gamma>0$ small enough such that 
for all $\ve\in (0, \ve_0]$ and for all $k=1, \dots, N_0$ it holds
\begin{equation*}
F_\ve (x, v)   - f(z_k)  \geq   \gamma         \   \  \  \    \ 
   \forall x\in  \Omega_\ve :=  B_{3\delta}(z_k) \setminus B_{\delta/2} \text{ and }
   \forall  v\in \mc N. 
\end{equation*}
It follows then from~\eqref{minmaxWittF}, the uniform bound $\ve^2|\nabla_\ve \chi_k|\leq 2$ 
and the existence of a $\tilde K>0$ with $\ve^d |\Omega_\ve| \leq \tilde K$ that
\begin{equation}\label{minmaxWittF2}
     \lp H_\ve (\chi_k  e^{- f/(2\ve)}) , \chi_k   e^{- f/(2\ve) } \rp_{\ell^2(\ve\mb Z^d)} \leq 
     2^{d+1}  \tilde K  e^{-[f(z_k)+\gamma]/\ve}  .
     \end{equation}
Putting together~\eqref{minimaxnormF} and~\eqref{minmaxWittF2}
gives the claim with e.g. $C=\gamma/2$ and $\ve_0\in( 0, \ve_0']$ sufficiently small. 
\end{proof}

\noindent
In the remainder of this section we provide the proofs of the basic estimates on $V_\ve$ given in Lemma~\ref{lowerboundV*F} and Lemma~\eqref{localV*F}.

\begin{proof}[Proof of Lemma~\ref{lowerboundV*F}]
A Taylor expansion gives for every $x\in \mb R^d$ the representation
\begin{equation}\label{TaylorUepsF}   V_\ve(x)         =          2 \sum_{v\in \mc N}     \sinh^2    \tfrac{\nabla f(x) \cdot v}{4}      +         \ve           \sum_{v\in \mc N}    e^{-\tfrac {\nabla f(x) \cdot v}{2}}  R_\ve(x,v)          ,    
\end{equation}
where, thanks to the boundedness of $\hess f$,
\[      \exists R>0   \   \   \text{s.t.}      \   \     |R_\ve(x,v)|    \leq R \        \  \forall   x \in \mb R^d,  v\in \mc N   \text{ and } \ve\in (0,1]          .   \]
In fact, one may write				
\begin{gather*}
V_\ve(x)        =   
  \sum_{v\in \mc N}   \left [ e^{-\tfrac {\nabla f(x) \cdot v}{2}}        -1 \right]        +         \ve       \sum_{v\in \mc N}   
     e^{-\tfrac {\nabla f(x) \cdot v}{2}}     \   
  \tfrac{1}{\ve} \left[   e^{-  \tfrac{f(x+\ve v) - f(x) -  \ve \nabla f(x) \cdot v }{2\ve}}   -       1          \right] ,     
\end{gather*}
and, using $\cosh 2t - 1   =     2\sinh^2 t$, 
\begin{gather*}
   \sum_{v\in \mc N}   \left [ e^{-\tfrac {\nabla f(x) \cdot v}{2}}        -1 \right]      =  
  \sum_{v\in \mc N}   \left [ \tfrac 12 e^{\tfrac {\nabla f(x) \cdot v}{2}}     +     
  \tfrac 12 e^{-\tfrac {\nabla f(x) \cdot v}{2}}           -         1 \right]        =  \\  
    \sum_{v\in \mc N}   \left [\cosh    \tfrac{\nabla f(x) \cdot v}{2}    -1 \right]    =   
    2 \sum_{v\in \mc N}   \sinh^2    \tfrac{\nabla f(x) \cdot v}{4}         .     
\end{gather*}
Moreover, for
\begin{gather*}
R_\ve(x,v)       :=       \tfrac{1}{\ve} \left[   e^{-  \tfrac{f(x+\ve v) - f(x) -  \ve \nabla f(x) \cdot v }{2\ve}}   -       1          \right], 
\end{gather*}
using $|e^{t}-1| \leq |t|e^{|t|}$ with 
\[   t      :=         -  \tfrac{f(x+\ve v) - f(x) -  \ve \nabla f(x) \cdot v }{2\ve}        ,     \]
and noting that,  due to the boundedness of $\hess f$, there exists a constant $A>0$ such that       
\[   |t|      \leq               \tfrac \ve4   \sup_x | \hess f(x) v\cdot v |   
 \leq   \     \tfrac \ve4   A |v|^2                     ,     \]
 one gets for $\ve  \in (0,1]$ and every $x\in \mb R^d$
 \begin{equation}\label{hessrF}  |R_\ve(x,v)|    \   \leq    \      \tfrac { A |v|^2}{4}   e^{\tfrac { A|v|^2  }{4}    }     \     \leq     \   
   \tfrac { \max_{v\in \mc N}A|v|^2  }{4}        e^{\tfrac { \max_{v\in \mc N}A|v|^2  }{4}    }     =:      R     >        0         .         
   \end{equation}
It follows from~\eqref{TaylorUepsF} and~\eqref{hessrF} that for $\ve\in (0,1]$ and every $x\in \mb R^d$
\begin{gather*}  V_\ve(x)          \geq          \sum_{v\in \mc N}   \left [\cosh    \tfrac{\nabla f(x) \cdot v}{2}    -1 \right]        -       \ve  R        \sum_{v\in \mc N}    e^{-\tfrac {\nabla f(x) \cdot v}{2}}     \
=    \\
=        \sum_{v\in \mc N}   \left [   (1 - \ve R )  \,   ( \cosh    \tfrac{\nabla f(x) \cdot v}{2}    -1)
   -       \ve R \right]        .
\end{gather*}
Using that    $\cosh t -1  \geq    t^2$  with $t=  \tfrac{\nabla f(x) \cdot v}{2}$ and
$ \sum_{v\in \mc N}    |\nabla f(x) \cdot v|^2         =     2  |\nabla f(x)|^2  $ we get  
for $\ve\in (0, \min\{1,\tfrac 1 R\})$ and every $x\in \mb R^d$ the lower bound
\begin{gather*}\label{estimateU1}  V_\ve(x)        \geq          \left[     \frac{(1-   \ve R)    
 }{2}   s |\nabla f(x)|^2     -            \ve R    \right]        .
\end{gather*}  
In particular
\[      V_\ve(x)       \geq              \frac{a^2}{2}       -         \ve   \,
 \left( \frac{ R a^2}{2}   +   R \right)  \   \  \     \   \  \    
      \forall x\in S   \   \text{ and  }   \   \forall \ve \in (0, \min\{1,\tfrac 1 R\})        .    \]
 The claim follows by chosing an $\ve_0 \in (0,  \min\{1, \tfrac 1R,   \frac{a^2}{Ra^2 + 2R} \})$ and $C =    \frac{a^2}{2}     \  -    \     \ve_0   \,
 \left( \frac{ R a^2}{2}   +   R \right)    $.
\end{proof}

\begin{proof}[Proof of Lemma~\ref{localV*F}]
This follows from a straightforward Taylor expansion. Indeed,
fixing $z\in \mb R^d$ such that $\nabla f(z)=0$ and $R>0$, 
we have on $B_R(z)$ the uniform estimate 
\begin{gather*}
-\tfrac{1}{2\ve} \left[ f(\cdot + \ve v) - f\right]   =    -  \tfrac 12 \nabla f \cdot v    - 
   \tfrac \ve 4 \hess f v \cdot v +  \mc O(\ve^2) .
\end{gather*}
Using the inequality $|e^t-1-t| \leq \tfrac 12 t^2 e^{|t|}$ with $t =  \tfrac \ve 4 \hess f \, v \cdot v +  \mc O(\ve^2) $ then gives
\begin{gather*}
V_{\ve}     =    
 \sum_{v\in \mc N}    \left\{  e^{-\tfrac 12 \nabla f \cdot v}    -  1          - 
 e^{-\tfrac 12 \nabla f\cdot v} 
\tfrac {\ve}{4} \hess f \, v \cdot v      +      \mc O(\ve^2)                  \right\}      =   \\  
 \sum_{v\in \mc N}    \left\{  \cosh{[\tfrac 12 \nabla f \cdot v]}    -  1            -  
 \cosh{[\tfrac 12 \nabla f \cdot v]} 
\tfrac {\ve}{4} \hess f \, v \cdot v      +      \mc O(\ve^2)                  \right\}     .   
\end{gather*}
The expansion $\cosh x = 1 + \frac 12 x^2 + \mc O(x^4)$ and the equalities 
$\sum_{v} |\nabla f\cdot v|^2 =   2|\nabla f|^2$ and $\sum_{v} \hess f \, v \cdot v = 
2\Delta f$ give 
\begin{gather*}
V_{\ve}     =  
 \frac 14 |\nabla f|^2      +    \mc O( \sum_{k} |\partial_k f|^4  )            -  
\tfrac {\ve}{2} \Delta f  +        \mc O(\ve |\nabla f|^2)     +     \mc O(\ve^2)     . 
\end{gather*}
Expanding all terms in $x$ around $z$, which gives in particular $|\nabla f (x)|^2 = [\hess f(z)]^2 (x-z)\cdot( x-z) + \mc O(|x-z|^3) $ and 
$\Delta f(x) = \Delta f (z) +   \mc O(|x-z|)$,   finishes the proof.  
\end{proof}

\section{Proof of Theorem~\ref{spectralgapF}}  \label{SectionMainProof2}

\subsection{General strategy}

\

\noindent
In order to compute the precise asymptotics of the smallest non-zero eigenvalue $\lambda(\ve)$ of $H_\ve$ we shall consider a suitable choice of an $\ve$-dependent test function $\psi_\ve$. The latter will be referred to as quasimode and its precise construction
will be given in Subsection~\ref{ssquasimodedefF}. Since $\psi_\ve$ will be chosen orthogonal to the ground state $e^{-f/(2\ve)}$ for every $\ve$, the upper bound
on $\lambda(\ve)$ given in Theorem~\ref{spectralgapF} will follow immediately from the Max-Min principle, giving
\begin{equation} \label{UBminimaxlambda}       \lambda(\ve)     \leq   \frac{\lp H_\ve \psi_\ve, \psi_\ve\rp_{\ell^2(\ve\mb Z^d)}}{\|\psi_\ve\|^2_{\ell^2(\ve\mb Z^d)}}    ,          
\end{equation}
and from the precise computation of the right hand side in the above formula by using the 
Laplace asymptotics on $\ve\mb Z^d$ given in Subsection~\ref{sectiondiscreteLaplace}. The result of these computations is the content of Proposition~\ref{PropNorm} and Proposition~\ref{propDirichlet}. 

\

\noindent
The proof of the lower bound on $\lambda(\ve)$ given in Theorem~\ref{spectralgapF} is more subtle. We shall derive it as a corollary of Theorem~\ref{main1F}
and the following abstract estimate, which was used in~\cite{DGLP} in a similar way. 
\begin{proposition} \label{abstractPROP} 
Let $(T, \mc D(T))$ be a nonnegative selfadjoint operator on a Hilbert space $(X,\lp \cdot, \cdot\rp)$. Moreover let $\tau>0$ and $P = \mbf 1_{[0,\tau]}(T)$ be the spectral projector of $T$ corresponding to the interval $[0,\tau]$ and let $\lambda = \sup ([0,\tau]\cap \spec(T))$. Then for every normalized $u\in \mc D(T)$ with 
$\lp T u, u\rp \neq 0 $ it holds 
\begin{equation*}
\lambda    \geq    \lp  T u, u     \rp    \left(  1   -    R(u)\right) , 
\end{equation*}
where $R(u) \geq 0 $ satisfies      
\[         [R(u)]^2    =    \tau^{-1}     \frac{\lp Tu, Tu \rp}{   \lp  T u, u \rp}     .  \] 
\end{proposition}

\begin{proof} We denote by $\|\cdot\|$ the Hilbert space norm and fix a $u\in \mc D(T)$
such that $\|u\|=1$ and $\lp T u, u\rp \neq 0 $. Then we have the estimates 
\begin{gather*}
\lambda    \geq      \|Pu\|^2   \lambda    =  \lp  u, \lambda   Pu\rp    \geq 
\lp  u, T Pu\rp    =  \\
   \lp  T u, u\rp    \left( 1    -     \frac{\lp Tu, u - Pu \rp}{   \lp  T u, u \rp}   \right)     \geq      \lp  T u, u\rp
   \left( 1    -     \frac{ \| Tu\| \| u - Pu\| }{   \lp  T u, u \rp}   \right)   . 
\end{gather*}
The claim follows now from the estimate  
\[    \| u - Pu\|^2      \leq       \tau^{-1}   \lp Tu, u   \rp  ,        \]
which is a consequence of the spectral theorem. 
\end{proof}
\noindent
We shall apply Proposition~\ref{abstractPROP} to the case  $T = H_\ve$, $\tau = 
C\ve$, where $C$ is the constant appearing in Theorem~\ref{main1F} and 
$u=   (\|\psi_\ve\|_{\ell^2(\ve\mb Z^d)})^{-1} \psi_\ve $, where $\psi_\ve$ is the same quasimode  used for the upper bound on $\lambda(\ve)$. By 
Theorem~\ref{main1F} (ii) we thus obtain a lower bound on $\lambda(\ve)$. 
The fact that this lower bound coincides with the lower bound given in Theorem~\ref{spectralgapF} is a 
consequence of the precise computation of the right hand side of~\eqref{UBminimaxlambda}, 
which we already mentioned (see Prop.~\ref{PropNorm} and Prop.~\ref{propDirichlet}), and
the estimate 
\begin{equation*}  (C\ve)^{-1}     \frac{\lp H_\ve \psi_\ve, H_\ve \psi_\ve \rp_{\ell^2(\ve\mb Z^d)}}{   \lp  H_\ve \psi_\ve, \psi_\ve \rp_{\ell^2(\ve\mb Z^d)}}       =          \mc O(\ve)    .   
\end{equation*}
The latter estimate will be a consequence of Proposition~\ref{propDirichlet}
and Proposition~\ref{PropCarre}, which is proven again by analyzing the Laplace asymptotics 
of a sum over $\ve\mb Z^d$.

\

\noindent
We shall assume throughout the rest of this section that the Assumption~\ref{H2} is satisfied.

\subsection{Definition of the quasimode $\psi_\ve$}\label{ssquasimodedefF}

\

\noindent
Let $s_1,\dots, s_n$ be the relevant saddle points of $f$, i.e. the critical points of index one of $f$ appearing in formula~\eqref{prefactor} 
defining the prefactor $A$. 
Given $x\in \mb R^d$ we associate to it a linear ``reaction coordinate'' $\xi_k= \xi_k(x)$ around the saddle point $s_k$, which parametrizes the unstable direction of $\hess f(s_k)$.
More precisely, we chose
one of the two normalized eigenvectors corresponding to the only negative eigenvalue
$\mu(s_k)$ of $\hess f(s_k)$, denote it by $\tau_k$, and set
\begin{equation}\label{definitionreaction}
     \xi_k(x)   =    \lp x- s_k , \tau_k \rp              \   \  \   \  \   \forall k=1,\dots, n  .      \end{equation}
Recalling our notation 
$\mc S_f(h) =    f^{-1} \left(  (-\infty, h) \right)$ for the    open sublevel set of $f$ corresponding to the height $h\in\mb R$, we consider for $\rho>0$ and $k=1, \dots, n$ the closed set 
\[     \mc R_k  =       \left\{   x\in  \overline{\mc S_f(h^*+ \rho)}: |\xi_k(x)| \leq \rho   \right \}    ,   \]
and the open set $\mc B = \mc S_f(h^*+\rho)   \setminus \left( \bigcup_k \mc R_{k} \right)$.

\noindent
Henceforth the parameter $\rho>0$ appearing in the definition of $\mc R_k$
and $\mc B$ is fixed sufficiently small such that
  the following properties  hold:
\begin{itemize}
\item[-]
the set $\mc B$ has exactly 
two connected components $\mc B^{(0)}$
and $\mc B^{(1)} $,  containing respectively $m_0$ and $m_1$.
\item[-]  $\mc R_{k}$ is disjoint from  $\mc R_{k'}$ for $k\neq k'$. 
\item[-]  For each $k=1, \dots, n$ the function $\vp_k   \   :  =   \   f   \   +   \    |\mu(s_k)| \xi_k^2 $ satisfies $\vp_k(x) > f(s_k) $ for every $x\in \mc R_{k}
\setminus \{s_k\}$.  
\end{itemize}
Note that $\hess \vp_k (s_k) =   |\hess f(s_k)|$. In other terms the quadratic approximation of $\vp_k$ around $s_k$ is obtained from that of $f$ by flipping the sign 
of the only negative eigenvalue of $\hess f(s_k)$.

\

\noindent
Let $\ve\in(0, 1]$.   
The quasimode $\psi_\ve$ for the  spectral gap is defined as follows.
We define first on the sublevel set  $\mc S_f( h^* + \rho)$ 
\[     \kappa_{\ve}(x)     =      \begin{cases}      +1      &    \text{ for }      x\in
\mc B^{(1)} ,        \\
-1       &      \text{ for }        x\in
\mc B^{(0)} ,    \\
C_{k,\ve}  \int_0^{\xi_k(x)}   \,  \chi( \eta) \, e^{- \frac{  |\mu(s_k)| \eta^2}{2\ve}}   \, d\eta  
 \    &      \text{ for }     x\in  \bigcup_{k}\mc R_{k}         .    
\end{cases}
   \]
The constant $C_{k,\ve}$ appearing above is defined as 
 \[     C_{k,\ve}       :=         \left[  \   \frac 12    \  \int_{-\infty}^{\infty} \chi( \eta) \, e^{- \frac{ |\mu(s_k)| \eta^2}{2\ve}}   \, d\eta    \right]^{-1},    \] 
 and $\chi \in C^\infty(\mb R;[0,1])$ satisfies $\chi\equiv 1$ on $[- \frac{\rho}{3},\frac\rho 3]$, $\chi(\eta)= 0$ 
 for $|\eta|\geq \frac 23 \rho$ and $\chi(\eta) =\chi(-\eta)$. Note that   
 \begin{equation}\label{cepsF}     
 \exists \gamma>0   \  \text{ such that }   \      C_{k,\ve}        =          2  \sqrt{  \frac{|\mu(s_k)|}{2\pi\ve}}   \    \left(  1 +   \mc O(e^{- \frac{\gamma}{\ve}}) \right)       .      
 \end{equation}
Note also that for each $k=1, \dots n$ the sign of the vector $\tau_k$ defining $\xi_k$ (see~\eqref{definitionreaction}) can be chosen 
such that 
$\kappa_\ve$ is $C^\infty$ on $\mc S_f( h^* + \rho)$, which we shall assume in the sequel.     
In order to extend $  \kappa$ to a smooth function defined on the whole $\mb R^d$
we introduce another cutoff function $\theta\in C^\infty(\mb R^d;[0,1])$ by setting for 
$x\in \mb R^d$     
\[     \theta(x)    =             \begin{cases}      1      &    \text{ for }      x \in \mc S_f(h^*+\frac \rho 2)     \\
0      &      \text{ for }     x\in \mb R^d \setminus  \mc S_f(h^* +  \frac 34 \rho)      \\
      \end{cases}.    
\]
Finally we define the quasimode $\psi_\ve$ by setting for $x\in \mb R^d$
\begin{equation}\label{defquasimodeF}   
 \psi_{\ve}(x)       =      
\left(  \frac 12 \theta(x)        \, \kappa_\ve(x)      -             \frac 12 \tfrac{ \lp  \theta    \kappa_{\ve} , e^{- f/\ve}   \rp_{\ell^2(\ve\mb Z^d)}   }{   \|e^{- f/(2\ve)}\|^2_{\ell^2(\ve\mb Z^d)} }     \right)  
  e^{- f/(2\ve)}      . 
\end{equation}
Note that $\psi_\ve \in C^\infty(\mb R^d)$ with compact support. In particular its restiction to 
$\ve\mb Z^d$, which we still denote by $\psi_\ve $, is in $C_c(\ve \mb Z^d) \subset {\textrm Dom}(V_\ve) $.  
Moreover, it follows from its very definition that $\psi_\ve$ is orthogonal 
to the ground state $e^{-f/(2\ve)}$ with respect to the scalar product $\lp \cdot, \cdot \rp_{\ell^2(\ve \mb Z^d)}$.

\subsection{Quasimode estimates}

\

\noindent
We now state the crucial estimates concerning the quasimode $\psi_\ve$. 
The proofs follow from straightforward computations exploiting the results of Subsection~\ref{sectiondiscreteLaplace} on the Laplace asymptotics for sums over $\ve\mb Z^d$.
We shall give the details in Subsection~\ref{ssectionproofsquasimode}.

\begin{proposition}
\label{PropNorm}
Assume ~\ref{H2} and let $\ve\in (0,1]$.
The function $\psi_\ve$ defined in~\eqref{defquasimodeF} satisfies
\begin{equation*}\label{propnorm}
 \|     \psi_\ve  \|^2_{\ell^2(\ve\mb Z^d)}        =      (2\pi\ve)^{\frac d2} J      
      e^{- h_*/\ve}    \left(   1     +     \mc O(\sqrt{\ve})       \right)        ,
\end{equation*}
where $h_*=\min\{f(m_0),f(m_1)\}$ is the minimum of  $f$ and 
\[    J   =  \begin{cases}    \left( \left(\det \hess f(m_1) \right)^{\frac 12}     +         \left(\det \hess f(m_0 )   \right)^{\frac 12}   \right)^{-1}     &       \text{ if } f(m_0) = f(m_1) ,  \\
     \left(\det \hess f(m_0 )   \right)^{-\frac 12}     &       \text{ if } f(m_0) < f(m_1)       .    
\end{cases}  \]
\end{proposition}

\begin{proposition}\label{propDirichlet}
Assume ~\ref{H2} and let $\ve\in (0,1]$.
The function $\psi_\ve$ defined in~\eqref{defquasimodeF} satisfies
\[      \lp H_\ve \psi_\ve, \psi_\ve \rp_{\ell^2(\ve\mb Z^d)} 
 =         \ve \sum_{k=1}^n \frac{|\mu(s_k)|}{2\pi}  \frac{(2\pi\ve)^{\frac d2}}{|\det \hess f(s_k)|^{\frac 12}}      e^{-h^*/\ve }   
\left(        1    +      
\mc O(\sqrt\ve)    \right)     , \]
where $\mu(s_k)$ is the only negative eigenvalue of $\hess f(s_k)$ and
$h^*$  is defined 
in~\eqref{heightofenergyF}.
\end{proposition}

\begin{proposition}\label{PropCarre}  
Assume ~\ref{H2} and let $\ve\in (0,1]$.
The function $\psi_\ve$ defined in~\eqref{defquasimodeF} satisfies
\[     \|  H_\ve \psi_\ve \|^2_{\ell^2(\ve\mb Z^d)}         =        \mc O(   \ve^3 ) \ e^{-h^*/\ve}        ,    \]
where $h^*$  is defined 
in~\eqref{heightofenergyF}.
\end{proposition}

\begin{remark}
With the stronger assumption $f\in C^4(\mb R^d)$ the $\mc O(\sqrt{\ve})$ error terms 
appearing in Proposition~\ref{PropNorm} and Proposition~\ref{PropNorm} can be shown to be actually 
$\mc O(\ve)$. Indeed it is enough to apply Proposition~\ref{GenPropositionLaplaceF}
with $k=4$ instead of $k=3$ each time it is used in the proofs given below. 
\end{remark}

\subsection{Proofs of the quasimode estimates}\label{ssectionproofsquasimode}

\begin{proof}[Proof of Proposition~\ref{PropNorm}]
Let $\ve\in (0,1]$.
We first consider the case $f(m_0) <  f(m_1)$. Then there exist $\alpha, \delta>0$ such that
$f \geq f(m_0)+\alpha $ on $[B_\delta(m_0)]^c$ and $\theta\kappa_\ve \equiv -1$ on $B_\delta(m_0)$. 
It follows that, denoting for short by $\Omega_\ve^\delta$ the bounded set
$[B_\delta^\ve(m_0)]^c\cap\supp(\theta)$, it holds
\begin{gather*}
 \|e^{- f/(2\ve)}\|^2_{\ell^2(\ve\mb Z^d)}    =   \\   
   \ve^d\sum_{x\in B_\delta^\ve(m_0)} e^{-f(x)/\ve}   +    \  
  e^{-[f(m_0) +\alpha]/\ve} \ve^d\sum_{x\in \Omega_{\ve}^\delta} e^{-[f(x)-\alpha] /\ve}
  =   \\
  \ve^d\sum_{x\in B_\delta^\ve(m_0)} e^{-f(x)/\ve}    \left( 1 + \mc O(e^{-[f(m_0) +\alpha]/\ve})    \right). 
\end{gather*}
Proposition~\ref{GenPropositionLaplaceF} gives then
\begin{gather}\label{NumFF}
 \|e^{- f/(2\ve)}\|^2_{\ell^2(\ve\mb Z^d)}    =   (2\pi\ve)^{\frac d2} 
     \left(\det \hess f(m_0) \right)^{-\frac 12}   e^{-f(m_0)}    \left( 1 + \mc O(\sqrt{\ve})    \right). 
\end{gather}
The same arguments and the estimate $|\theta\kappa_\ve| \leq 1$ show that
\begin{gather}\label{denomFF}
\lp  \theta    \kappa_{\ve} , e^{- f/\ve}   \rp_{\ell^2(\ve\mb Z^d)} 
=    -    (2\pi\ve)^{\frac d2} 
     \left(\det \hess f(m_0) \right)^{-\frac 12}   e^{-f(m_0)}    \left( 1 + \mc O(\sqrt{\ve})    \right). 
\end{gather}
Taking the quotient between~\eqref{NumFF} and~\eqref{denomFF} it 
follows then from the definition of $\psi_\ve$ that 
\begin{equation*} 
 \psi_{\ve}(x)       =      
\left(  \frac 12 \theta(x)        \, \kappa_\ve(x)      +             \frac 12 
+    \mc O(\sqrt{\ve})     \right)  
  e^{- f/(2\ve)}      . 
\end{equation*}
The norm   $\psi_\ve  \|^2_{\ell^2(\ve\mb Z^d)}$ can be now computed by splitting again the 
sum in two sums, respectively over  $B_\delta^\ve(m_0)$ and $\Omega_\delta^\ve$. The 
conclusion in the case $f(m_0)< f(m_1)$ follows by again using Proposition~\ref{GenPropositionLaplaceF} for the first sum 
and arguing as above  for the second sum. 

\

\noindent
We now consider the case $f(m_0)= f(m_1)$. It follows from the definition of  $\psi_\ve$ that
\begin{equation}\label{variancedecompositionf} \|     \psi_\ve  \|^2_{\ell^2(\ve\mb Z^d)}   
         =       \frac 14  \|    \theta \kappa e^{- f/(2\ve)}\|^2_{\ell^2(\ve\mb Z^d)}           -         
\frac 14   \frac{ \lp  \theta    \kappa , e^{- f/\ve}   \rp^2_{\ell^2(\ve\mb Z^d)}   }
{   \|e^{- f/(2\ve)}\|^2_{\ell^2(\ve\mb Z^d)} }                  .     
\end{equation}
Let  $\alpha, \delta>0$ such that
$f \geq f(m_0)+\alpha $ on $[B_\delta(m_0) \cup
B_\delta(m_1)   ]^c$ and $\theta\kappa_\ve \equiv -1$ on $B_\delta(m_0)$,
$\theta\kappa_\ve \equiv 1$ on $B_\delta(m_1)$. 
With arguments as above one gets 
 \begin{gather*}
 \|e^{- f/(2\ve)}\|^2_{\ell^2(\ve\mb Z^d)}   =  
 \|    \theta \kappa e^{- f/(2\ve)}\|^2_{\ell^2(\ve\mb Z^d)}    \left(    1       +     \mc O(\sqrt{\ve})         \right)           =  \\
 =      
       (2\pi\ve)^{\frac d2} 
     \left[ \left(\det \hess f(m_1) \right)^{-\frac 12}    +       \left(\det \hess f(m_0 )   \right)^{-\frac 12}   \right]     
      e^{- f(m_1)/\ve}    \left(   1  +   \mc O(\sqrt{\ve})        \right)         , 
\end{gather*}
\begin{gather*}
      \lp  \theta    \kappa , e^{- f/\ve}   \rp_{\ell^2(\ve\mb Z^d)}          =  \\
=       
    (2\pi\ve)^{\frac d2} 
     \left[ \left(\det \hess f(m_1) \right)^{-\frac 12} -  \left(\det \hess f(m_0 )   \right)^{-\frac 12}   \right]     
      e^{- f(m_1)/\ve}      \left(   1   +    \mc O(\sqrt{\ve})        \right)   . 
\end{gather*}
Putting these expressions into~\eqref{variancedecompositionf}, the desired result~\eqref{propnorm} follows after some algebraic manipulations. 
\end{proof}

\begin{proof}[Proof of Proposition~\ref{propDirichlet}]
Let $\ve\in (0,1]$. Using~\eqref{Witten*} and the notation $F_\ve(x,v) =
\frac 12 [ f(x) + f(x+\ve v) ] $ gives 
\begin{equation*}
  \lp H_\ve \psi_\ve, \psi_\ve \rp_{\ell^2(\ve\mb Z^d)}    =   
       \tfrac{\ve^2}{4} \|e^{-F_\ve/(2\ve)}\nabla_\ve (\theta \kappa_\ve) \|^2_{\ell^2(\ve \mb Z^d; \mb R^\mc N)}    . 
\end{equation*}
Since the function $\theta$ has support in $\mc S_f(  h^* +  \frac 34 \rho)$, we can restrict (for $\ve$ sufficiently small) the sum running  
over $ \ve \mb Z^d$ to the bounded set $\ve \mb Z^d \cap \mc S_f(  h^* +   \rho)$.
Note that $\mc S_f(  h^* +   \rho)$ is the union of the disjoint sets
 $ \mc B$ and $\bigcup_k (\mc R_k  \setminus [\mc S_f(  h^* +   \rho)]^c)$. \footnote{The subtraction of 
 $[\mc S_f(  h^* +   \rho)]^c$ is necessary to have disjoint sets, but is not really relevant, since it concerns only boundary terms which do not matter in the computations given below.} 
  We write in the sequel for short
$ \mc R_{k,\ve} := \ve \mb Z^d    \cap     (\mc R_{ k} 
 \setminus [\mc S_f(  h^* +   \rho)]^c )$      and  $\mc B_\ve :=  \ve \mb Z^d    \cap  \mc B  $
and discuss below separately
the sum over $\cup_{k=1}^n   \mc R_{k, \ve} $, which will give the main contribution,
and the sum over $\mc B_\ve$,  which will give a negligible contribution.

\

\noindent
Below we shall use the Taylor expansion
\begin{equation}\label{TaylorgF}   
e^{-F_\ve(x,v)/\ve}      =      e^{- f(x)/\ve}     e^{- \nabla f(x)\cdot v / 2}  
     \left(       1       +      \mc O(\ve)  \right)  .    
\end{equation}

\

{      \it 1) Analysis on  $\cup_{k=1}^n   \mc R_{k, \ve} $.}

\, 

\noindent
In order to get rid of $\theta$ we take $\delta>0$ small enough such that  for each $k$ 
it holds   $B^\ve_\delta(s_k) \subset  \mc  R_{k,\ve}    \cap \mc S_f(  h^* +  \frac 14 \rho)  \subset \mc R_{k,\ve}$.   Since
$\theta$ and $\kappa_\ve$ are uniformly bounded in $\ve$ and 
$f\geq h^* + \frac 14 \rho$ 
on $\mc R_{k,\ve} \setminus B^\ve_\delta(s_k) $, we get
using~\eqref{TaylorgF} that, for $\ve>0$ sufficiently small (and thus also for 
$\ve\in (0,1]$), 
it holds 
\begin{gather}      \frac{\ve^d}{2}
           \sum_{x\in \mc R_{k,\ve} }  \sum_{v\in \mc N}    \frac 14\left[\theta\kappa_\ve \, (x+\ve v) -  \theta\kappa_\ve(x) \right]^2 e^{-F_\ve(x,v)/\ve}     
               =     \nonumber \\
                  \frac{\ve^d}{2}
           \sum_{x\in B_\delta^\ve(s_k) }  \sum_{v\in \mc N}    \frac 14\left[\kappa_\ve \, (x+\ve v) -  \kappa_\ve(x) \right]^2    e^{-F_\ve(x,v)/\ve}
              +          \mc O(e^{- (h^* + \frac \rho 4)})    ,        \label{fromRtoR'F} 
           \end{gather}
where we have used also that for $\ve$ sufficiently small $\theta(x)=\theta(x+\ve v)=1$
for $x\in B_\delta^\ve(s_k)$.  

\

\noindent
We discuss now in detail the behavior of $x \mapsto \kappa_\ve \, (x+\ve v) -  \kappa_\ve(x) $ near $s_k$. 
For $k =1, \dots, n$ and  $x\in \mc R'_{k}$, $v\in \mc N$ and $\ve\in (0,1]$  consider the function 
$G=G_{k,x,v,\ve}:[0,1] \to \mb R$ defined by 
\begin{equation} \label{defF}    
 G(\delta)      =            C_{k,\ve}^{-1}       \left[\kappa_\ve \, (x+\delta v) -  \kappa_\ve(x) \right]    =   
    \int_{\xi_k(x)}^{\xi_k(x+\delta v)}  \chi( \eta) \, e^{-   |\mu(s_k)| \eta^2/ (2\ve)}   \, d\eta        .     
    \end{equation}
Note that $G(0)=0$ ,  
 $ G'(0)    =         e^{-   |\mu(s_k)| \xi_k^2(x)/(2\ve)}          \chi(\xi_k(x))  \tau_k \cdot v   $,      
\[    G''(0)     =     e^{-  |\mu(s_k)| \xi_k^2(x)/(2\ve) }    |\tau_k \cdot v |^2 
\left[   \chi'(\xi_k(x))  -    |\mu(s_k)|\frac{\xi_k(x)}{\ve}    \chi(\xi_k(x)) \right]  ,   \]
and for every $\delta \in [0,1]$ 
\begin{gather*} \ve^3 G'''(\delta)     =   \\
 \ve e^{  -  |\mu(s_k)| \xi_k^2(x+\delta v )/(2\ve) }   
  (\tau_k \cdot v )^3 \left[  |\mu(s_k)|^2 \xi_k^2(x+\delta v) \chi(\xi_k(x+\delta v)  
  +   \ve R\right] ,  
  \end{gather*}
  where $R$ is not depending on $\ve$ and bounded in $k,x,v$. 
By Taylor expansion it follows that 
\begin{gather} 
  G(\ve)     =       \ve     e^{-   |\mu(s_k)| \xi_k^2(x)/(2\ve)}   \chi(\xi_k(x))     \,  \times    \label{TaylorF} 
     \\
         \left[   \tau_k \cdot v   -  \frac12  |\mu(s_k)| \, \xi_k(x) 
   |\tau_k \cdot v |^2     +     \mc O(|x-s_k|^2) \right]    
       \left(        1      +   
\mc O(\ve)    \right)              .       \nonumber        
\end{gather}
It follows from~\eqref{defF},~\eqref{TaylorF},~\eqref{TaylorgF},~\eqref{cepsF} and the two identities
\begin{equation}  \label{identitycosh}
\sum_{v\in \mc N}    |   \tau_k \cdot v|^2     e^{-\nabla f(x)\cdot v/2}      =      2            
 \sum_{j=1}^d  ( e_j\cdot\tau_k )^2    \cosh \tfrac{\partial_j f(x) }{2},   
\end{equation}
\begin{equation*}   \sum_{v\in \mc N}    (\tau_k \cdot v )^3  \, e^{-\nabla f(x) \cdot v/2}   
    =       -     2\sum_j   (e_j \cdot \tau_k)^3 \sinh \tfrac{\partial_j f(x)}{2}   
\end{equation*}
that for $k =1, \dots, n$ and $x\in \mc R'_{k}$ and $\ve>0$ small enough 
\begin{gather} 
 \frac 12  \sum_{v\in \mc N}    \frac 14\left[\kappa_\ve \, (x+\ve v) -  \kappa_\ve(x) \right]^2 
  e^{-F_\ve(x,v)/\ve}    =    \label{alphaformF}    \\
  \ve     e^{- f(s_k)/\ve}\frac{|\mu(s_k)|}{2\pi}  \  e^{-  \vp_k(x)/\ve}      \alpha_k(x)     \left(        1      +    
\mc O(\ve)       +      \mc O\left(|x-s_k|^2\right)      \right)      ,   \nonumber  
   \end{gather}
  where for shortness   we have set    $ \vp_k(x)      =     f(x)    - f(s_k)  +       |\mu(s_k)| \xi_k^2(x)   $    and
  \begin{gather*}    \alpha_k(x)    =   \\
      \chi^2(\xi_k(x))    \  \sum_{j=1}^d    \left\{
  (e_j\cdot\tau_k)^2    \cosh \frac{\partial_j f(x) }{2}          +       |\mu(s_k)|  \xi_k(x)  
   ( e_j\cdot\tau_k)^3 \sinh \frac{\partial_j f(x)}{2}   \right\}   =   \\
   1 +     \mc O(|x-s_k|^2)      .  
  \end{gather*}
Putting together~\eqref{fromRtoR'F}, \eqref{alphaformF}, using Proposition~\ref{GenPropositionLaplaceF}, summing over $k$ and the fact 
that $f(s_k) = h^*$ for every $k$ finally gives   
 \begin{gather*}
  \frac{\ve^d}{2}
           \sum_{x\in \cup_{k}\mc R_{k,\ve}  }  \sum_{v\in \mc N}    \frac 14\left[\kappa_\ve \, (x+\ve v) -  \kappa_\ve(x) \right]^2    e^{-F_\ve(x,v)/\ve}     =   \\
            \ve \sum_{k=1}^n \frac{|\mu(s_k)|}{2\pi}  \frac{(2\pi\ve)^{\frac d2}}{|\det \hess f(s_k)|^{\frac 12}}    \  e^{-\frac{h^*}{\ve} }    
\left(        1    +   
\mc O(\sqrt{\ve})    \right)   .    
 \end{gather*}

\

{      \it 2) Analysis on  $\mc B_\ve $.}

\, 

\noindent
As in Step 1) we get rid of $\theta$ by considering the set  $\mc B_\ve' =  \mc  B_\ve    \cap \mc S_f(  h^* +  \frac 14 \rho)  \subset \mc B$.   
Arguing as before and now using  that $\kappa_\ve(x)= \kappa_\ve(x+\ve v)$ for every $x\in \mc B_\ve'$, $v\in \mc N$ and $\ve$ sufficiently small, gives
then 
\begin{gather*}      \frac{\ve^d}{2}
           \sum_{x\in \mc B_{\ve} }  \sum_{v\in \mc N}    \frac 14\left[\theta\kappa_\ve \, (x+\ve v) -  \theta\kappa_\ve(x) \right]^2  e^{-F_\ve(x,v)/\ve}
               =     \nonumber \\
                  \frac{\ve^d}{2}
           \sum_{x\in \mc B'_{\ve} }  \sum_{v\in \mc N}    \frac 14\left[\kappa_\ve \, (x+\ve v) -  \kappa_\ve(x) \right]^2  e^{-F_\ve(x,v)/\ve} 
                +          \mc O(e^{- (h^* + \frac \rho 4)})    \\
                   =         \mc O(e^{- (h^* + \frac \rho 4)})  .                 \end{gather*}
\end{proof}

\begin{proof}[Proof of Proposition~\ref{PropCarre}]
The isomporphism~\eqref{conjugation} gives the identity
\begin{gather*} 
  \|  H_\ve \psi_\ve \|^2_{\ell^2(\ve\mb Z^d)}    = 
    \| \ve \Phi_\ve \left[  L_{\ve}   \Phi^{-1}_\ve[\psi]  \right]   \|^2_{\ell^2(\rho_\ve)}      =     \\ 
                   \tfrac{\ve^{d}}{4}
           \sum_{x\in \ve \mb Z^d}   
               \left(       \sum_{v\in \mc N}  e^{-\tfrac 12 \nabla_\ve f(x,v)}
               \ve \nabla_\ve  (\theta\kappa_\ve)  (x, v)          \right)^2
             e^{-\frac{f(x)}{\ve}}     . 
\end{gather*}
Since the function $\theta$ has support in $\mc S_f(  h^* +  \frac 34 \rho)$, we can restrict (for $\ve$ sufficiently small) the sum   
over $ \ve \mb Z^d$ to the bounded set $\ve \mb Z^d \cap \mc S_f(  h^* +   \rho)$.
As in the proof of Proposition~\ref{propDirichlet} we shall split the latter into the
disjoint sets $\cup_k  \mc R_{k,\ve}$,with $ \mc R_{k,\ve} := \ve \mb Z^d    \cap     (\mc R_{ k} 
 \setminus [\mc S_f(  h^* +   \rho)]^c )$, and  $\mc B_\ve :=  \ve \mb Z^d    \cap  \mc B  $.
 
 \

 \noindent
We discuss here in detail only the contribution coming from the sets $ \mc R_{k,\ve}$. Indeed the 
 sum over $\mc B_\ve$ can be neglected arguing exactly as in Step 2) of Proposition~\ref{propDirichlet} and using 
instead of~\eqref{TaylorgF} that by Taylor expansion 
\begin{equation}\label{TaylorrF}  
e^{-\tfrac 12 \nabla_\ve f(x,v) }  =        e^{- \nabla f(x)\cdot v/2}       \left(       1      +      \mc O(\ve)  \right)          .    
\end{equation}

 \

{      \it Analysis on  $\cup_{k=1}^n   \mc R_{k, \ve} $.}

\, 

\noindent
As in the proof of Proposition~\ref{propDirichlet} we first 
get rid of $\theta$ by taking a $\delta>0$ small enough such that  for each $k$ 
it holds   $B^\ve_\delta(s_k) \subset  \mc  R_{k,\ve}    \cap \mc S_f(  h^* +  \frac 14 \rho)  \subset \mc R_{k,\ve}$.   Since
$\theta$ and $\kappa_\ve$ are uniformly bounded in $\ve$ and 
$f\geq h^* + \frac 14 \rho$ 
on $\mc R_{k,\ve} \setminus B^\ve_\delta(s_k) $, we get
using~\eqref{TaylorrF} that, for $\ve>0$ sufficiently small (and thus also for 
$\ve\in (0,1]$), 
it holds 
\begin{gather}
   \tfrac{\ve^{d}}{4}
           \sum_{x\in \mc R_{k,\ve}}   
               \left(       \sum_{v\in \mc N}  e^{-\tfrac 12 \nabla_\ve f(x,v)}
               \ve \nabla_\ve  (\theta\kappa_\ve)  (x, v)          \right)^2
             e^{- f(x)/\ve}  
               =     \label{ridoftheta2}     \\
           \tfrac{\ve^{d}}{4}
           \sum_{x\in B_\delta^\ve(s_k)  }    
              \left(       \sum_{v\in \mc N}  e^{-\tfrac 12 \nabla_\ve f(x,v)}
               \ve \nabla_\ve  \kappa_\ve  (x, v)          \right)^2
             e^{- f(x)/\ve}        +          \mc O(e^{- (h^* + \frac \rho 4)})           .    \nonumber
\end{gather}
A computation already used in the proof of Proposition~\ref{propDirichlet} (see~\eqref{TaylorF})
yields
\begin{gather*}
       e^{   |\mu(s_k)| \xi_k^2(x)/(2\ve)}     C^{-1}_{k,\ve} 
       \, \ve \nabla_\ve \kappa_\ve  (x, v)  =  \\
       \ve      \chi(\xi_k(x))       \left[   \tau_k \cdot v   -  \frac12  |\mu(s_k)| \, \xi_k(x) 
   |\tau_k \cdot v |^2     +     \mc O(|x-s_k|^2) \right]    
       \left(        1      +   
\mc O(\ve)    \right)         . 
\end{gather*}
Hence, using~\eqref{cepsF},~\eqref{TaylorrF}, the identity~\eqref{identitycosh} and the identity 
\begin{equation*}   \sum_{v\in \mc N}    \tau_k \cdot v  
e^{- \nabla f(x) \cdot v/2}      =       -    
 2\sum_j   e_j\cdot\tau_k \sinh \tfrac{\partial_j f(x)}{2}          ,
\end{equation*}
one obtains 
\begin{gather}  
e^{   |\mu(s_k)| \xi_k^2(x)/(2\ve)}  
       \sum_{v\in \mc N}  e^{-\tfrac 12 \nabla_\ve f(x,v)}
               \ve \nabla_\ve  \kappa_\ve  (x, v)      
                 =  
\sqrt \ve   \alpha_k(x)   \left(        1      +   
\mc O(\ve)    \right)        ,        \label{defalpha2F}
\end{gather}
 with
 \begin{gather}    \alpha_k(x)       =    -  \sqrt{  \tfrac{2 |\mu(s_k)|}{\pi}}     \chi(\xi_k(x))      \,  \times         \nonumber \\
          \sum_{j=1}^d      \left[   2  e_j\cdot\tau_k \sinh \tfrac{\partial_j f(x)}{2}         
    +         |\mu(s_k)|  \xi_k(x) ( e_j\cdot\tau_k )^2    \cosh \tfrac{\partial_j f(x) }{2}   + 
      \mc O(|x-s_k|^2)   
        \right]        .     \label{defalpha}
\end{gather}
Observing that 
  \begin{gather*}         \sum_{j=1}^d        2  e_j\cdot\tau_k \sinh \tfrac{\partial_j f(x)}{2}         =       
  \lp \hess f(s_k) \tau_k, x-s_k\rp     +     \mc O(|x-s_k|^2)    =  \\ 
    -|\mu(s_k)|    \, \xi_k (x)      +     \mc O(|x-s_k|^2)            ,     
    \end{gather*}
  and that
      \begin{gather*}         \sum_{j=1}^d   
      ( e_j\cdot\tau_k )^2    \cosh \tfrac{\partial_j f(x) }{2}    =    1 +  \mc O(|x-s_k|^2) 
  \end{gather*}
shows that the first order terms in~\eqref{defalpha} cancel out and thus 
 $ \alpha_k(x)  =   \mc O(|x-s_k|^2)$. It follows then from~\eqref{ridoftheta2},~\eqref{defalpha2F} that
there exists a constant $C>0$ such that for every $\ve\in (0,1]$ and every $k=1, \dots, n$
\begin{gather*}
 \tfrac{\ve^{d}}{4}
           \sum_{x\in \mc R_{k,\ve}}   
               \left(       \sum_{v\in \mc N}  e^{-\tfrac 12 \nabla_\ve f(x,v)}
               \ve \nabla_\ve  (\theta\kappa_\ve)  (x, v)          \right)^2
             e^{- f(x)/\ve}      \leq     \\
             C     \ve^{d+1} e^{-f(s_k)/\ve} \sum_{x\in  B_\delta^\ve(s_k)}   |x-s_k|^4  e^{-\vp_k(x)/\ve}  
             =     \mc O(     \ve^3) e^{- h^*/\ve}     , 
           \end{gather*}
with $ \vp_k(x)      =     f(x)    +       |\mu(s_k)| \xi_k^2(x)  - f(s_k)  $ and 
with the last estimate following from Proposition~\ref{GenPropositionLaplaceF} by taking 
$m=2$.    
\end{proof}

\noindent{\bf Acknowledgements:}  The author gratefully acknowledges the financial support of  HIM Bonn in the framework of the 2019 Junior Trimester Programs ``Kinetic Theory'' and ``Randomness, PDEs and Nonlinear Fluctuations''.

\end{document}